%% file: sn-article.tex
\RequirePackage{amsthm}
\documentclass[sn-mathphys-num]{sn-jnl}

\usepackage{graphicx}%
\usepackage{multirow}%
\usepackage{amsmath,amssymb,amsfonts}%
\usepackage[capitalise]{cleveref}
\usepackage{amsthm}%
\usepackage{mathrsfs}%
\usepackage[title]{appendix}%
\usepackage{xcolor}%
\usepackage{textcomp}%
\usepackage{manyfoot}%
\usepackage{booktabs}%
\usepackage{algorithm}%
\usepackage{algorithmicx}%
\usepackage{algpseudocode}%
\usepackage{listings}%

\theoremstyle{thmstyleone}%
\newtheorem{theorem}{Theorem}
\newtheorem{lemma}[theorem]{Lemma}%

\newtheorem{corollary}[theorem]{Corollary}%

\theoremstyle{thmstyletwo}%

\theoremstyle{thmstylethree}%
\newtheorem{definition}{Definition}%

\input{preamble}

\begin{document}

\title[Combining Crown Structures for Vulnerability Measures]{Combining Crown Structures for Vulnerability Measures}


\author[1]{\fnm{Katrin} \sur{Casel}}\email{katrin.casel@hu-berlin.de}
\equalcont{These authors contributed equally to this work.}

\author[2]{\fnm{Tobias} \sur{Friedrich}}\email{Tobias.Friedrich@hpi.de}
\equalcont{These authors contributed equally to this work.}

\author[2]{\fnm{Aikaterini} \sur{Niklanovits}}\email{Aikaterini.Niklanovits@hpi.de}
\equalcont{These authors contributed equally to this work.}

\author[2]{\fnm{Kirill} \sur{Simonov}}\email{Kirill.Simonov@hpi.de}
\equalcont{These authors contributed equally to this work.}

\author*[2]{\fnm{Ziena} \sur{Zeif}}\email{Ziena.Zeif@hpi.de}
\equalcont{These authors contributed equally to this work.}

 
\affil[1]{\orgname{Humboldt-Universität zu Berlin},\country{Germany}}

\affil[2]{\orgname{Hasso Plattner Institute, University of Potsdam},  \country{Germany}}



\abstract{Over the past decades, various metrics have emerged in graph theory to grasp the complex nature of network vulnerability.
In this paper, we study two specific measures: (weighted) vertex integrity (wVI) and (weighted) component order connectivity (wCOC).
These measures not only evaluate the number of vertices that need to be removed to decompose a graph into fragments, but also take into account the size of the largest remaining component.
The main focus of our paper is on kernelization algorithms tailored to both measures.
We capitalize on the structural attributes inherent in different crown decompositions, strategically combining them to introduce novel kernelization algorithms that advance the current state of the field.
In particular, we extend the scope of the balanced crown decomposition provided by Casel et al.~\cite{DBLP:conf/esa/Casel0INZ21} and expand the applicability of crown decomposition techniques.

In summary, we improve the vertex kernel of VI from $p^3$ to $3p^2$, and of wVI from $p^3$ to $3(p^2 + p^{1.5} \pl)$, where $\pl < p$ represents the weight of the heaviest component after removing a solution.
For wCOC we improve the vertex kernel from $\O(k^2W + kW^2)$ to $3\mu(k + \sqrt{\mu}W)$, where $\mu = \max(k,W)$.
We also give a combinatorial algorithm that provides a $2kW$ vertex kernel in fixed-parameter tractable time when parameterized by $r$, where $r \leq k$ is the size of a maximum $(W+1)$-packing. 
We further show that the algorithm computing the $2kW$ vertex kernel for COC can be transformed into a polynomial algorithm for two special cases, namely when $W=1$, which corresponds to the well-known vertex cover problem, and for claw-free graphs.
In particular, we show a new way to obtain a $2k$ vertex kernel (or to obtain a 2-approximation) for the vertex cover problem by only using crown structures.}

\keywords{Crown Decomposition, Kernelization, Vertex Integrity, Component Order Connectivity} 


\maketitle

\section{Introduction}
\input{intro}

\section{Preliminaries}\label{sec::prelim}
 \input{prelim}

\paragraph*{Crown Decompositions}\label{sec::prelimCrownAndPack}
\input{prelimCrownDecomposition}

\section{Improved Kernels for VI, wVI and wCOC}\label{sec::VIandwVI}

The vertex integrity of a graph models finding an optimal balance between removing vertices and keeping small connected parts of a graph.
As a reminder: In the formal definition, a graph $G$ and a parameter $p$ are given.
The task is to find a vertex set $S \subseteq V$ such that $|S| + \max_{Q \in \comp(G - S)} |Q|$ is at most $p$. 
In the weighted vertex integrity problem (wVI), i.e.~the vertices have weights, the aim is that $w(S) + \max_{Q \in \comp(G - S)} w(Q)$ is at most $p$.
We improve the vertex kernel of $p^3$ provided by Drange et al.~\cite{DBLP:journals/algorithmica/DrangeDH16} for  both variants VI and wVI.
To obtain such a kernel, they essentially established that vertices $v \in V$ with $w(v) + w(N(v)) > p$ belong to a solution if we have a yes-instance in hand.
A simple counting argument after removing those vertices provides a $p^3$ vertex kernel.
We improve this by employing crown decompositions, which offer valuable insights into the structural properties of vertex sets rather than focusing solely on individual vertices.
Directly applying this method to problems like COC, as seen in prior works such as \cite{DBLP:conf/esa/Casel0INZ21,DBLP:journals/tcs/ChenFSWY19,DBLP:conf/iwpec/KumarL16,DBLP:conf/tamc/XiaoK17}, is challenging because we lack prior knowledge about the size of remaining components after solution removal.
However, if we had at least a lower bound, we could theoretically safely reduce our instance accordingly.
To establish such a bound, we engage in an interplay between packings and separators, where the balanced crown decomposition (BCD) proves instrumental.
By identifying reducible structures in the input instance through embedding a \dbe into BCD after determining a suitable bound, we prove the following theorems regarding VI and wVI.

\begin{theorem}
	\label{theorem::kernelVerInt}
	The vertex integrity problem admits a vertex kernel of size $3p^2$ in time $\O\left(\log(p)|V|^4|E|\right)$.
\end{theorem}

\begin{theorem}
	\label{theorem::kernelWeightVerInt}
	The weighted vertex integrity problem admits a vertex kernel of size $3(p^2 + p^{1.5} \pl)$ in time $\O\left(\log(p)|V|^4|E|\right)$, where $\pl$ is at most the size of the largest component after removing a solution.
\end{theorem}

Closely related to wVI is the weighted component order connectivity problem (wCOC).
Given  a vertex-weighted graph $G=(V,E,w)$ and two parameters $k,W \in \N$, the task is to find a vertex set $S \subseteq V$ such that $w(S) \leq k$ where each component weighs at most $W$.
The techniques employed to derive the kernel for wVI can be seamlessly applied to wCOC, thereby enhancing the current state of the art vertex kernel of $kW(k+W)+k$.
This kernel is also provided by Drange et al.~\cite{DBLP:journals/algorithmica/DrangeDH16} in a similar way as for wVI.

\begin{theorem}
	\label{theorem::wCOC}
	The weighted component order connectivity problem admits a vertex kernel of size $3\mu(k + \sqrt{\mu}W)$, where $\mu = \max(k,W)$.
	Furthermore, such a kernel can be computed in time $\O\left(r^2 k |V||E|\right)$, where $r$ is the size of a maximum $(W+1)$-packing.
\end{theorem}

Before we prove these theorems, we give some notations that we use in this section.
We define them for the weighted case, as the unweighted case can be viewed in the same way with unit weights.
We say that \emph{$S$ is a solution} if it satisfies $w(S) + \max_{Q \in \comp(G-S)} w(Q) \leq p$ for wVI or $w(S) \leq k$ and $\max_{Q \in \comp(G-S)} w(Q) \leq W$ for wCOC.
We denote an instance of wVI by $(G,p)$ and an instance of wCOC by $(G,k,W)$. 
We say that $(G,p)$ or $(G,k,W)$ is a \emph{yes-instance} if there is a solution, otherwise, we say that it is a \emph{no-instance}.
Let $\mathcal{S} \subseteq 2^{V(G)}$ be all solutions for $(G,p)$.
We define $\pl := \min_{S \in \mathcal{S}} \left(\max_{Q \in \comp(G-S)} w(Q)\right)$ which is the optimum lower bound  on the size of the connected components after the removal of any solution; where for no-instances we set $\pl=p$.
For all instances $(G,p)$ of wVI, we assume  $w(v) < p$ for every $v \in V$ and that $G$ contains a connected component of weight more than $p$.
Lastly, for a better readability, we will consistently assume that when computing a $\lambda$-BCD, we  disregard small components, meaning that if  $(C,H,\calR,f)$ is a $\lambda$-BCD for $G$, then $C\cup H\cup V(\calR)$ are only the vertices that are contained in some component of size (or weight) more than $\lambda$ of $G$. 

\subsection{Vertex Integrity}\label{sec::VI}
\input{vertex_integrity_kernel}
\subsection{Weighted Vertex Integrity and Component Order Connectivity}\label{sec:wVIandwCOC}
\input{weighted_vertex_intergrity_problem}

\section{Kernels for Component Order Connectivity}\label{sec::COC}
\input{W-separator-kernel-results}


\bibliography{literature}
\appendix
\section{Proof of Theorem~\ref{thm::dbExapnsion}}\label{app:a}
\input{journalappendix}
\end{document}

%% file: preamble.tex
\usepackage{graphicx}
\usepackage{algpseudocode}
\usepackage{dsfont}
\usepackage{mathtools}
\usepackage{tikz-cd}
\usetikzlibrary{decorations.pathmorphing}
\usepackage[size=tiny]{todonotes}
\usepackage{todonotes}
\usepackage{xspace}

\usepackage{xparse} 

\usepackage{dsfont}

\usepackage{thmtools}

\NewDocumentCommand{\mathOrText}{m}
{%
	\ensuremath{#1}\xspace%
}

\DeclareDocumentCommand{\Pr}{m o}
{
	\mathOrText{ \mathds{P}\left[ #1 \IfNoValueF{#2}{\;\middle\vert\; #2} \right]}
}
\DeclareDocumentCommand{\E}{m o}
{
	\mathOrText{ \mathds{E}\left[ #1 \IfNoValueF{#2}{\;\middle\vert\; #2} \right]}
}
\DeclareDocumentCommand{\filtration}{o}
{
	\mathOrText{ \mathcal{F}\IfNoValueF{#1}{_{#1}}}
}

\usepackage{todonotes}

\def\N{\mathds{N}}

\def\O{\mathcal{O}}
\def\opt{\textsc{OPT}}

\def\P{\mathcal{P}}

\def\opt{\text{OPT}}

\def\Q{\mathcal{Q}}
\def\B{\mathcal{B}}

\def\ora{\overrightarrow}

\def\comp{\mathbb{CC}}
\def\calR{\mathfrak{R}}
\def\calV{\mathcal{V}}
\def\chrf{(C,H,\calR,f)}
\def\sepw{\texttt{sep}_W}
\def\argSepw{\texttt{arg\hspace{0.1cm}sep}_W}
\def\gLarge{G^{>W}}

\def\pl{p_\ell}
\def\pu{p}

\def\dbe{\normalfont{DBE}\xspace}
\def\cd{\normalfont{CD}\xspace}
\def\wmax{w_{\max}}

\def\ubc{y}
\def\lbhc{x}
\def\alg2kw{\textsc{AlgCOC}\xspace}
\def\2alg2kw{\textsc{AlgCOC-2}\xspace}
\def\algVI{\textsc{AlgVI}\xspace}
\def\algWVI{\textsc{AlgWVI}\xspace}
\def\algWCOC{\textsc{AlgWCOC}\xspace}

%% file: intro.tex

In the study of graph theory different scales have emerged over the past decades to capture the complex nature of network vulnerability.
While the main focus is on the connectivity of vertices and edges, subtle aspects of vulnerability such as the number of resulting components, the size distribution of the remaining components, and the disparity between them are becoming increasingly interesting~\cite{DBLP:conf/esa/BentertHK23,DBLP:journals/tcs/GimaHKKO22,DBLP:conf/walcom/GimaHKM0O24,DBLP:journals/corr/abs-2402-09971,DBLP:journals/algorithmica/GimaO24,DBLP:conf/isaac/LampisM21}.
Our focus in this study is on two specific ways to measure vulnerability: (weighted) vertex integrity and (weighted) component order connectivity.
These measures not only evaluate the number of vertices that need to be removed to decompose a graph into fragments, but also take into account the size of the largest remaining component.
Incorporating these aspects provides a more comprehensive understanding of network resilience.

Informally, \emph{vertex integrity} (VI) is a model for the balance between removing few vertices and keeping small connected parts of a graph.
More formally, given a graph $G=(V,E)$ and a number $p \in \N$, the task for VI is to find a set of vertices $S \subseteq V$ such that $|S|$ plus the size of the largest component when removing $S$ from $G$ is at most $p$.
In the vertex weighted version (wVI), the goal is bounding the total weight of the removed vertices plus the weight of the heaviest component by $p$.
This problem was introduced by Barefoot et al.~\cite{barefoot1987vulnerability} as a way to measure vulnerability of communication networks.
Recently, it has drawn the interest of the parameterized complexity community due to its status as a natural parameter that renders numerous NP-hard problems amenable to fixed parameter tractability (FPT).
This means that for these problems, solutions can be computed within a time frame represented by $f(p) \cdot n^{\O(1)}$, where $f$ is a computable function \cite{DBLP:conf/isaac/LampisM21}.
It is interesting to see how vertex integrity relates to other well-known measures of network structure.
It imposes greater constraints compared to metrics such as treedepth, treewidth, or pathwidth, as the vertex integrity of a graph serves as an upper bound for these parameters.
However, it encompasses a wider range of scenarios compared to vertex cover, where a vertex cover of a graph is an upper bound for its vertex integrity.
This makes it a key in understanding how to efficiently solve problems in the world of network analysis.

The measure \emph{component order connectivity} (COC) can be seen as the refined version of VI.
Given a graph $G=(V,E)$ and two parameters $k,W \in \N$, the goal of COC is to remove $k$ vertices such that each connected component in the resulting graph has at most $W$ vertices --- also known in the literature as the \emph{$W$-separator problem} or \emph{$\alpha$-balanced separator problem}, where $\alpha  \in (0,1)$ and $W=\alpha|V|$.
In the vertex-weighted version (wCOC), the goal is to remove vertices of total weight at most $k$ such that the weight of the heaviest remaining component is at most $W$.
An equivalent view of this problem is to search for the minimum number of vertices required to cover or hit every connected subgraph of size $W+1$.
In particular, $W=1$ corresponds to covering all edges, showing that the COC is a natural generalization of the vertex cover problem.

The focus of the paper is on kernelization algorithms tailored for both weighted and unweighted versions of VI and COC when parameterized by $p$ and $k+W$, respectively.
Kernelization algorithms can be thought of as formalized preprocessing techniques aimed to reduce insances of optimization problems.
Of particular interest in this work are crown decompositions, which are generally used as established structures for safe instance reduction --- where ``safe'' means that any optimal solution to the reduced instance can efficiently be transformed into an optimal solution of the original one.
Essentially, a crown decomposition partitions the vertex set into three disjoint components: the crown, the head, and the body.
Here, the head acts as a separator between the crown and the body.
This structural arrangement becomes useful when specific relationships between the head and the crown are required.
Such relationships ultimately enable us to shrink instances by eliminating these designated parts from the graph.
The properties of this structural layout, coupled with its existence depending on the instance size, enable the development of efficient kernelization algorithms for different problem domains.
Notably, crown decompositions have also recently found utility in approximation algorithms for graph packing and partitioning problems~\cite{DBLP:conf/esa/Casel0INZ21}.
For further exploration of crown decompositions, including their variations and applications, we recommend the comprehensive survey paper by Jacob et al.~\cite{DBLP:journals/algorithms/JacobMR23}.

Our methods take advantage of the structural characteristics found in various crown decompositions, leading to the development of new kernelization algorithms that improve the state of the art.
In essence, this work expands upon the applications of the balanced crown decomposition introduced by Casel et al.~\cite{DBLP:conf/esa/Casel0INZ21}, specifically by integrating different crown decompositions into this framework.

\paragraph*{Related Work}
Recently, vertex integrity (VI) has received considerable attention as a structural graph parameter \cite{DBLP:conf/esa/BentertHK23,DBLP:journals/tcs/GimaHKKO22,DBLP:journals/algorithmica/GimaO24,DBLP:conf/isaac/LampisM21}. Some results in the literature are stated in terms of the \emph{fracture number} of the graph~\cite{DBLP:journals/algorithmica/GanianOR21}, which is a closely related parameter that takes the maximum of the modulator size and the largest remaining component; in particular, fracture number is always within a factor of two from vertex integrity. Gima et al.~\cite{DBLP:conf/walcom/GimaHKM0O24} conducted a systematic investigation into structural parameterizations of computing the VI and wVI which was further extended by Hanaka et al.~\cite{DBLP:journals/corr/abs-2402-09971}.
Additionally, there are notable results concerning special graph classes~\cite{DBLP:journals/dam/BaggaBGLP92,clark1987computational,DBLP:journals/algorithmica/DrangeDH16,DBLP:journals/dam/KratschKM97,DBLP:journals/ijcm/LiZZ08}.
In our context, regarding related work on VI and wVI, Fellows and Stueckle presented an algorithm that solves the problem in time $\O(p^{3p} n)$~\cite{fellows1989immersion}.
This was subsequently improved by Drange et al.~\cite{DBLP:journals/algorithmica/DrangeDH16}, even for the weighted case, to $\O(p^{p+1}n)$.
In the same paper, they presented the first vertex-kernel of size $p^3$ for both VI and wVI.

Considering, COC and wCOC, it is unlikely that kernelization algorithms for these problems can be achieved by considering $k$ or $W$ alone in polynomial time.
Indeed, $W=1$ corresponds to the NP-hard vertex cover problem, which shows that $W$ (alone) is not a suitable parameter.
For the parameter $k$, the problem is $W[1]$-hard even when restricted to split graphs~\cite{DBLP:journals/algorithmica/DrangeDH16}.
These lower bounds lead to the study of parameterization by $k+W$.
The best known algorithm with respect to these parameters finds a solution in time $n^{O(1)} \cdot 2^{\O(\log(W) \cdot k)}$~\cite{DBLP:journals/algorithmica/DrangeDH16}.
Unless the exponential time hypothesis fails, the authors prove that this running time is tight in the sense that there is no algorithm that solves the problem in time $n^{\O(1)} \cdot 2^{o(\log(W) \cdot k)}$.
The best known approximation algorithm has a multiplicative gap guarantee of $\O(\log(W))$ to the optimal solution with a running time of $n^{\O(1)} \cdot 2^{\O(W)}$~\cite{DBLP:journals/corr/Lee16c}.
In~\cite{DBLP:journals/corr/Lee16c}, the authors also showed that the superpolynomial dependence on $W$ may be needed to achieve a polylogarithmic approximation.
Using this algorithm as a subroutine, the vertex integrity can be approximated within
a factor of $\O(\log(\opt))$, where $\opt$ is the vertex integrity.

Regarding kernelization algorithms, there is a sequence of results which successively improve the vertex-kernel of COC.
The first results came from Chen et al.~\cite{DBLP:journals/tcs/ChenFSWY19} and Drange et al.~\cite{DBLP:journals/algorithmica/DrangeDH16}, who provided  kernels of size in $\O(k W^3)$ and~$\O(k^2W + kW^2)$, respectively.
The result of Drange et al.~also holds for wCOC and is the only result for this case.
This was improved simultaneously by Xiao~\cite{DBLP:journals/jcss/Xiao17a} as well as by Kumar and Lokshtanov~\cite{DBLP:conf/iwpec/KumarL16} to a $\O(kW^2)$ kernel.
These works also provide the first $\O(kW)$ kernels, but with different constants and running times.
Kumar and Lokshtanov~\cite{DBLP:conf/iwpec/KumarL16} present a $2kW$ kernel in a running time of $n^{\O(W)}$ by using linear
programming (LP) methods with an exponential number of constraints.
The runtime can be improved to $2^{\O(W)} \cdot n^{\O(1)}$ as already mentioned in the book of Fomin et al.~\cite{fomin2019kernelization} (Section 6.4.2).
Roughly speaking, the idea is to use the ellipsoid method with separation
oracles to solve the linear program, where the separation oracle
uses a method called color coding to find violated constraints that makes it polynomial in $W$.
Note that if $W$ is a constant this $2kW$ kernel is a polynomial time kernel improving on some previous results.
This includes for instance the improvement of the $5k$ kernel provided by Xiao and Kou~\cite{DBLP:conf/tamc/XiaoK17} to a $4k$ kernel for the well-studied $P_2$-covering problem, where a $P_2$ is a path with 2 edges.
The first linear kernel in both parameters in polynomial time, i.e.~an $\O(kW)$ vertex kernel, is presented by Xiao~\cite{DBLP:journals/jcss/Xiao17a}, who provides a $9kW$ vertex kernel.
Finally, this was improved  by Casel et al.~\cite{DBLP:conf/esa/Casel0INZ21} to a~$3kW$ vertex kernel, which also holds for a more general setting.
Namely, to find $k$ vertices in a vertex weighted graph such that after their removal each component weighs at most~$W$.
Note that the weights of the separator, i.e.~the chosen $k$ vertices, play no role in this problem compared to wCOC.
With the exception of the $\O(k^2W + kW^2)$ vertex kernel of Drange et al.~\cite{DBLP:journals/algorithmica/DrangeDH16}, all achieved vertex kernels essentially use crown structures.

\paragraph*{Our Contribution}
The provided running times are based on an input graph $G=(V,E)$. 
We improve the vertex kernel for VI from $p^3$ to $3p^2$ in time $\O\left(\log(p)|V|^4|E|\right)$.
For wVI, we improve it to $3(p^2 + p^{1.5} \pl)$ in time $\O\left(\log(p)|V|^4|E|\right)$, where $\pl \leq p$ represents the weight of the heaviest component after removing a solution.

To better explain the results of COC and wCOC, consider the problem of a maximum $\lambda$-packing for $\lambda \in \N$.
Given a (vertex weighted) graph $G$, this problem aims to maximize the number of disjoint connected subgraphs, each of size (or weight) at least $\lambda$.
It is worth noting that the size of a maximum $(W+1)$-packing serves as a lower bound on the size of an optimal solution of COC (or wCOC), since each element in the packing must contain at least one vertex of it — in terms of linear programming, it is the dual of COC.

For wCOC we improve the vertex kernel of $\O(k^2W + kW^2)$ to $3\mu(k + \sqrt{\mu}W)$ in time $\O\left(r^2 k |V||E|\right)$, where $\mu = \max(k,W)$ and $r \leq |V|$ is the size of a maximum $(W+1)$-packing.
For the unweighted version, we provide a $2kW$ vertex kernel in an FPT-runtime of $\O(r^3|V||E| \cdot r^{\min(3r,k)})$, where $r \leq k $ is the size of a maximum $(W+1)$-packing.
Comparing this result with the state of the art, disregarding the FPT-runtime aspect, we improve upon the best-known polynomial algorithm, achieving a kernel of size $3kW$~\cite{DBLP:conf/esa/Casel0INZ21}.
A $2kW$ vertex kernel is also presented in \cite{DBLP:conf/iwpec/KumarL16}, albeit with an exponential runtime using linear programming methods, which, as mentioned, can be enhanced to an FPT-runtime regarding parameter $W$.
In contrast, our result is entirely combinatorial and has an FPT-runtime in the parameter of a maximum $(W+1)$-packing $r \leq k$.
It should be noted that, strictly speaking, the $2kW$ vertex kernel of \cite{DBLP:conf/iwpec/KumarL16} and our work cannot be considered a kernel, given that the runtime dependency is exponential with respect to the parameters $W$ and $k$, respectively. However, for the sake of simplicity, we refer to it as a kernel, with an explicit mention of the runtime dependency.
As previously stated, note that it is unlikely that an FPT-runtime will be sufficient for solving COC (or wCOC) when considering either $W$ or $k$ alone.
We further show that the algorithm computing the $2kW$ vertex kernel for COC can be transformed into a polynomial algorithm for two special cases.
The first case arises when $W=1$, i.e., for the vertex cover problem.
Here, we provide a new method for obtaining a vertex kernel of $2k$ (or obtaining a $2$-approximation) using only crown decompositions.
The second special case is for the restriction of COC to claw-free graphs.
Notably, it remains open whether there is a polynomial-time algorithm for COC on claw-free graphs, or whether the problem is NP-hard on this class; we defer this question to future work.

Regarding these special cases, until 2017, a $2k$ vertex kernel for $W=1$ was known through crown decompositions, albeit computed using both crown decompositions and linear programming. 
Previously, only a $3k$ vertex kernel was known using crown decompositions alone.
In 2018, Li and Zhu~\cite{DBLP:journals/tcs/LiZ18} provided a $2k$ vertex kernel solely based on crown structures.
They refined the classical crown decomposition, which possessed an additional property allowing the remaining vertices of the graph to be decomposed into a matching and odd cycles after exhaustively applying the corresponding reduction rule.
In contrast, our algorithm identifies the reducible structures, which must exist if the size of the input graph exceeds $2k$.
For claw-free graphs, we are not aware of previous work regarding kernels for VI or COC on this class. However, the classical result of Minty~\cite{DBLP:journals/jct/Minty80} states that one can find vertex cover of smallest weight in a claw-free graph in polynomial time.

Unfortunately, we were unable to transform the FPT-runtime kernelization algorithm into a polynomial-time algorithm in general.
Nevertheless, we believe that our insights into the structural properties of crown decompositions pave the way to achieving this goal.

%% file: prelim.tex
In the upcoming paragraphs, we discuss the standard terminology related to graphs and parameterized complexity. Additionally, we introduce some crown structures that are utilized throughout this paper.

\paragraph*{\textbf{Graph Terminology}}
Let $G=(V,E)$ be a graph.
We use $V(G)$ and $E(G)$ to denote $V$ and $E$, respectively.
We define the \emph{size} of a subgraph $G' \subseteq G$ as the number of its vertices, where we denote the size of $G$ by $n$.
For $v \in V(G)$ we denote its neighborhood by $N(v):=\{u\in V(G):uv\in E(G)\}$.
Let $V' \subseteq V$ be a vertex subset.
We define $G[V']$ as the induced subgraph of $V'$, $G-V' := G[V \setminus V']$ and $N(V') := \left(\bigcup_{v \in V'} N(v)\right) \setminus V'$.
For a vertex $v$ we define $G+v$ as the graph $G$ with an additional isolated vertex $v$.
Similarly, for a graph $H$ we define $G+H$ as the graph $G$ with an additional isolated graph $H$.
For sets of vertex sets $\mathcal{V} \subseteq 2^V$ we abuse notation and use $V(\mathcal{V}) = \bigcup_{Q \in \mathcal{V}} Q$ and $N(\mathcal{V}) := N(V(\calV)) = \left(\bigcup_{Q \in \mathcal{V}} N(Q)\right) \setminus V(\mathcal{V})$.
We define $\comp(G) \subset 2^V$ as the connected components of $G$ as vertex sets.
For vertex subsets $H, C \subset V$, we define a function $f \colon \comp(C) \to H$. For any subset $H' \subseteq H$, we frequently use the inverse $f^{-1}(H')$, which returns sets of vertices. Typically, however, we are interested in the union of these vertex sets, denoted $V(f^{-1}(H'))$. For readability, we may omit the vertex operator when it is clear from the context that we refer to the union.
We use the same terminology for \emph{vertex weighted graphs}, denoted by $G=(V,E,w)$ where $w$ is a weight function $w \colon V \to \N$.
Let $G=(V,E,w)$ be a vertex weighted graph.
For $V' \subseteq V$, $G' \subseteq G$ (meaning that $G'$ is a subgraph of $G$) and $\calV \subset 2^V$ we define $w(V')$, $w(G')$ and $w(G')$ as $\sum_{v \in V'} w(v)$, $\sum_{v \in V(G')} w(v)$ and $\sum_{v \in V(\calV)} w(v)$, respectively.

\paragraph*{\textbf{Parameterized Terminology}}
We use the standard terminology for parameterized complexity, which is also used, for example, in \cite{downey2012parameterized,fomin2019kernelization}.
A \emph{parameterized problem} is a decision problem with respect to certain instance parameters.
Instances of a parameterized problem are usually given as tuple $(I,k)$, where $k$ is the parameter.

If  there exists an algorithm that decides each instance $(I,k)$ of a parameterized problem $\Pi$ in time $f(k) \cdot |I|^c$, where $f$ is a computable function and $c$ is a constant, then $\Pi$  is called \emph{fixed-parameter tractable}.
We say $(I,k)$ is a \emph{yes-instance} if the answer to the decision problem is positive, otherwise we say $(I,k)$ is a \emph{no-instance}.

Of particular interest in this work are kernelizations, which can be roughly described as formalized preprocessing.
More formally, a polynomial algorithm is called a \emph{kernelization} for a parameterized problem $\Pi$  if it maps any instance $(I,k)$ of $\Pi$  to an instance $(I',k')$ of $\Pi$  such that $(I',k')$ is a yes-instance if and only if $(I,k)$ is a yes-instance, $|I'| \leq g(k)$, and $k' \leq g'(k)$ for computable functions $g,g'$.

%% file: prelimCrownDecomposition.tex

We present a more general variant of crown decomposition that also captures commonly used crown decompositions or expansions (strictly speaking the relevant parts of it), which we explain in a moment.
This novel decomposition can be easily derived from existing results of~\cite{DBLP:conf/esa/Casel0INZ21}, but to the best of our knowledge has never been used in this form before (see \cref{pic::MaxFlowNetwork2} (left) for an illustration).

\begin{definition}[Demanded balanced expansion and weighted crown decomposition] 
	\label{def::demBalDec}
	For a graph $G=(A \cup B,E,w)$,
	a partition $A_1 \cup A_2$ of $A$,
	a function $f \colon \comp(G[B]) \to A$,
	demands $D=\{d_a\}_{a \in A}$ with $d_a \in \N$ for each $a \in A$
	and $y \in \N$
	the tuple $(A_1,A_2,y,f,D)$ is a demanded balanced expansion \normalfont{(DBE)} if
	\begin{enumerate}
		\item
		$w(Q) \leq \ubc$ for each $Q \in \comp(G[B])$ 
		\item 
		$f(Q) \in N(Q)$ for each $Q \in \comp(G[B])$
		\item 
		$N(f^{-1}(A_1)) \subseteq A_1$
		\item
		\label{item::bDCD-ub}
		$w(a) + w\left(f^{-1}(a)\right)
		\begin{cases}
			 > d_a - \ubc + 1 \text{ for each } a \in A_1\\
			 \leq d_a + \ubc - 1 \text{ for each } a \in A_2
		\end{cases}$
	\end{enumerate}
	To simplify the notation we introduce two further special cases of a \dbe:
	\begin{itemize}
		\item If the demands are the same for each $a \in A$, e.g.~$d_a = x$ for every $a \in A$ with $x \in \N$, then we write only the value $x$ instead of a vector $D=\{d_a\}_{a \in A}$ in a \dbe-tuple, i.e.~$(A_1,A_2,y,f,D)=(A_1,A_2,y,f,x)$.
		\item
		For $q \in \N$ we call $(C,H,f)$ a ($q,\ubc$) crown decomposition (($q,\ubc$)-\cd) if it is a $(H, \varnothing,y,f,q+y-1)$ \dbe with $H=A$ and $C = B = f^{-1}(H)$.
		(This term simplifies the reference to the reducible structure of crown $C$ and head $H$.) 
	\end{itemize}

\end{definition}

\begin{figure}
	\begin{minipage}{0.49\textwidth}
		\begin{tikzpicture}[scale=0.8,y=0.80pt, x=0.80pt, inner sep=0pt, outer sep=0pt]
			\path[draw=blue, line width=0.55mm] (0,60)--(98,60);	
			\path[draw=blue, line width=0.55mm] (0,45)--(97,60);			
			\path[draw=blue, line width=0.55mm] (0,30)--(97,59);
			\path[draw=black, line width=0.25mm] (0,30)--(98,30);
			\path[draw=red, line width=0.55mm] (0,15)--(98,30);
			\path[draw=orange, line width=0.55mm] (0,0)--(98,-28);
			\path[draw=black, line width=0.25mm] (0,-30)--(98,58);
			\path[draw=green, line width=0.55mm] (0,-30)--(97,0);
			\path[draw=black, line width=0.25mm] (0,-15)--(98,30);
			\path[draw=orange, line width=0.55mm] (0,-15)--(97,-31);

			\path[fill=black!60!, line width=0.25mm] (0,60) circle (0.09cm);
			\path[fill=black!60!, line width=0.25mm] (0,45) circle (0.09cm);
			\path[fill=black!60!, line width=0.25mm] (0,30) circle (0.09cm);
			\path[fill=black!60!, line width=0.25mm] (0,15) circle (0.09cm);
			\path[fill=black!60!, line width=0.25mm] (0,0) circle (0.09cm);
			\path[fill=black!60!, line width=0.25mm] (0,-15) circle (0.09cm);			
			\path[fill=black!60!, line width=0.25mm] (0,-30) circle (0.09cm);
			\path[fill=black, line width=0.25mm] (100,60) circle (0.09cm);
			\path[fill=black, line width=0.25mm] (100,30) circle (0.09cm);
			\path[fill=black, line width=0.25mm] (100,0) circle (0.09cm);			
			\path[fill=black, line width=0.25mm] (100,-30) circle (0.09cm);
			
			\node at (0,75) {$B$};
			\node at (100,75) {$A$};
			\node at (115,45) {{\large $\Bigg \}$}};
			\node at (130,45) {$A_1$};
			\node at (115,-15) {{\large $\Bigg \}$}};
			\node at (130,-15) {$A_2$};
			\node at (50,-50) {};
		\end{tikzpicture} \hspace*{2.4cm}
	\end{minipage}
	\begin{minipage}{0.5\textwidth}
		\includegraphics[scale=0.2]{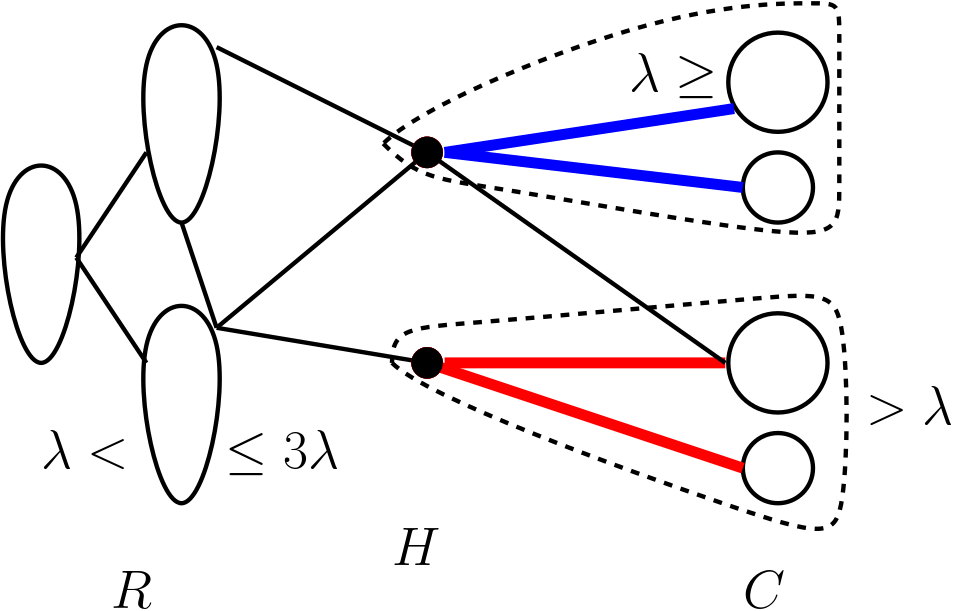}
	\end{minipage}

	\caption{\\
		\textbf{Left:} Let $A = \{a_1,a_2,a_3,a_4\}$ be ordered in a top down manner, $w(Q) = 1$ for every $Q \in \comp(G[B])$ and $w(a) = 1$ for every $a \in A$.
		Then, $(\{a_1,a_2\},\{a_3,a_4\},1,f,\{3,1,3,4\})$ is a \dbe, where the assignment $f$ are depicted with corresponding colored bold edges.\\
		\textbf{Right:} A $\lambda$-balanced crown decomposition, where the assignment $f$ are depicted with corresponding colored bold edges. The two dashed lines illustrate that $w(h) + w(f^{-1}(h)) > \lambda$ for every $h \in H$ while $w(Q) \leq \lambda$ for every $Q \in \comp(G[C])$.}
	\label{pic::MaxFlowNetwork2}
\end{figure}

Considering a graph instance $G=(V,E)$, e.g.~from VI or COC, Definition~\ref{def::demBalDec} usually describes only a subgraph where the reducible structure is sought, i.e.~$A,B$ are vertex subsets of $V$. 
It is crucial that $A$ already separates $B$ from $V \setminus B$.
Then, the vertex set $A_1 \subseteq A$ (cf. Definition~\ref{def::demBalDec}) typically represents the structure targeted for reduction and can be seen as the ``head'' of a crown decomposition.
The components within $\comp(G[B])$ assigned to $A_1$ (the separated part through $A_1$), denoted by $f^{-1}(A_1)$, form the ``crown''.
In the context of optimization problems such as the vertex cover problem, a successful reduction often involves the head being part of an optimal solution.
Subsequently, upon its removal, the crown --- separated by the head --- no longer needs to be considered.

Ideally, we wish to have $A_1=A$, but even when this is not the case the balanced part of a \dbe ensures that the elements mapped to $A_2$ are bounded, which finally allows us to bound the number of vertices (or the sum of the vertex weights) of $A_2 \cup V(f^{-1}(A_2))$.  


Consider a crown decomposition, which is a partition of the vertex set into body, head and crown, where the head separates the body from the crown.
The head and crown of the \emph{classical crown decomposition} corresponds to a $(1,1)$-CD, the \emph{$q$-expansion} for $q \in \N$ to a $(q,1)$-CD and the \emph{weighted crown decomposition} for $q \in \N$ to a $(q,\lbhc)$-CD.
The last known structure captured by Definition~\ref{def::demBalDec} is the \emph{balanced expansion} which corresponds for $\lbhc,\ubc \in \N$ to a $(A_1,A_2,y,f,x)$ \dbe.
The essential new structural property of a \dbe compared to a balanced expansion are the varying  demands for $A$.

Let $G=(A \cup B,E)$ be a graph.
For $A' \subseteq A$ we define $\B_{A'}$ as the components $Q \in \comp(G[B])$ with $N(Q) \subseteq A'$.
The following theorem, easily derived as modification of the results of~\cite{DBLP:conf/esa/Casel0INZ21}, gives the runtime to find a \dbe and the existence guarantee of $A_1$ depending on the size or weight in the graph. For the sake of completeness, we give the proof in Appendix~\ref{app:a}.

\begin{theorem}[Demanded balanced expansion]
	\label{thm::dbExapnsion}
	Let $G=\left(A \cup B, E, w \right)$ be a graph with no isolated components in $\comp(G[B])$, i.e.~every component of $\comp(G[B])$ contains at least one neighbor of $A$.
	Let $\ubc \geq \max_{Q \in \comp(G[B])} w(Q)$ and $D=\{d_a\}_{a \in A}$ demands with $d_a \in \N$ for each $a \in A$.
	A demanded balanced expansion ($A_1,A_2,y,f,D$) can be computed in $\mathcal{O}\left(|V|\,|E|\right)$ time.
	Furthermore, if there is an $A' \subseteq A$ with $w(A') + w(V(\B_{A'})) \geq \sum_{a \in A'} d_a$, then $A_1 \neq \varnothing$.	
\end{theorem}

The next graph structure that we use for our kernelization algorithms is introduced by Casel et al.~\cite{DBLP:conf/esa/Casel0INZ21} and is a combination of a balanced connected partition and a weighted crown decomposition, which is called \emph{balanced crown decomposition}.
Formally it is defined as follows (see also Figure~\ref{pic::MaxFlowNetwork2} (right) for an illustration).  
\begin{definition}[$\lambda$-balanced crown decomposition]\label{def:lccd}
	A $\lambda$-balanced crown decomposition of a graph $G=(V,E,w)$ is a tuple  $\chrf$, where $\{H,C,R\}$ is a partition of~$V$, the set $\mathfrak{R}$ is a partition of $R$, and $f\colon \comp(G[C])\rightarrow H$, such that:
	
	\noindent
	$
	\left.\parbox{0.65\textwidth}{
		\begin{itemize}
			\item[1.] \text{$w(Q) \leq \lambda$ for each  $Q\in\comp(G[C])$,}
			\item[2.] \text{$f(Q) \in N(Q)$ for each $Q\in\comp(G[C])$,}
			\item[3.] \text{$N(C) \subseteq H$,}  
			\item[4.]  \text{$w(h) + w(f^{-1}(h)) > \lambda$ for each $h\in H$ and}	
	\end{itemize}}
	\right\}$ \hfill  \emph{($\lambda,\lambda$)-CD}
	\begin{itemize}
		\item[5.] $G[R']$ is connected and $\lambda < w(R')\leq  3\lambda$ for each $R' \in \calR$. 	
	\end{itemize}
\end{definition}
 
We use $\lambda$-BCD as an abbreviation for a $\lambda$-balanced crown decomposition. 
Looking at the original definition of a $\lambda$-BCD in~\cite{DBLP:conf/esa/Casel0INZ21} (Definition 6), we shift the $\lambda$ value by one, which allows us to change several inequalities between strict and non-strict for a clearer representation, while still keeping the definition the same.
Furthermore, $\lambda$ must be at least two in the original definition, but to simplify the understanding of the application to the problems considered in this paper, it makes more sense that a $\lambda$-BCD can also exist with $\lambda=1$.

The authors in~\cite{DBLP:conf/esa/Casel0INZ21} provide an algorithm that finds a $\lambda$-BCD in polynomial time, as given in the following.

\begin{theorem}[Balanced crown decomposition theorem, \cite{DBLP:conf/esa/Casel0INZ21}~Theorem 7]\label{theorem:lccd}
	Let $G = (V,E,w)$ be a graph and $\lambda \in \N$, 
	such that each connected component in $G$ has weight larger than $\lambda$.
	A $\lambda$-balanced crown decomposition $(C,H,\calR,f)$ of $G$ can be computed in $\mathcal{O}\left(r^2\,|V|\,|E|\right)$ time, where $r=|H|+|\calR| < |V|$ is at most the size of a maximum $(\lambda+1)$-packing.
\end{theorem}

We end the preliminary section with a formal definition of the subgraph packing problem.
Given a graph $G=(V,E)$ and two parameters $r,\lambda \in \N$.
We say that $P_1, \dots, P_m \subseteq V$ is a $\lambda$-\emph{packing} if for all $i,j \in [m]$ with $i \neq j$ the induced subgraph $G[P_i]$ is connected, $|P_i| \geq \lambda$, and $P_i \cap P_j = \varnothing$.
The task is to find a $\lambda$-packing of size at least $r$.

%% file: vertex_integrity_kernel.tex

The rest of this section is dedicated to the proof of Theorem~\ref{theorem::kernelVerInt}.
 We first give the types of reductions we want to use, always with respect to some $c\leq \pl$, and then discuss how to efficiently compute an estimate of $c$ to apply these reduction rules.

\begin{lemma}
	\label{lemma::saveRedVerInt}
	Let $(G,p)$ be an instance of \normalfont{VI} where not all connected components are of size at most $p$ and let $(C,H,f)$ be a $(\pu,c)$-CD in $G$ for some  $c\leq \pl$, such that $N(C) \subseteq H$.
	Then, $(G,p)$ is a yes-instance if and only if $(G-(H \cup C)+K_c,p-|H|)$ is a yes-instance.
\end{lemma}
\begin{proof}
Let $G' = G-\{H \cup C\}+K_c$. 
First assume $(G,p)$ is a yes-instance, and let $S$ be an according solution  with $s = \max_{Q \in \comp(G - S)} |Q|$. Note that since not all components of $G$ have size at most $p$, it follows that $S\neq \emptyset$ and $s<p$.
	We claim that $S' = S \setminus (H \cup C)$ is a solution for $(G',p-|H|)$. It suffices to show that $\max_{Q \in \comp(G'-S')} |Q| \leq s$ and $|S'|\leq |S|-|H|$. More precisely, we can then conclude that $S'$ has size at most $|S'|\leq (p-|H|)-s$ and creates components of size at most $s$ when deleted from $G'$, which makes it a solution for $(G',p-|H|)$, meaning that $(G',p-|H|)$ is a yes-instance.
	
	For the claim $\max_{Q \in \comp(G'-S')} |Q| \leq s$, first note that $c\leq \pl$ by the choice of $c$, and that by the definition of $\pl$ we have $\pl\leq s$. Thus the $K_c$ component of $G'$ has size at most $s$. For the other components in $G'-S'$, suppose that $G-(H \cup C)-S'$ contains a connected component $Q'$ of size larger than $s$.
	Note that $V(Q') \subset V \setminus (H \cup C)$.
	Since we do not remove vertices from $S$ in $V \setminus (C \cup H)$ to build $S'$, we have that $Q'$ is also a component of $G-S$, contradicting $s= \max_{Q \in \comp(G-S)} |Q|$.	
		
For the claim $|S'| \leq |S| - |H|$ 	observe that  $\calV_H = \{\{h\} \cup V(f^{-1}(h))\}_{h \in H}$ forms a $\pu$-packing of size $|H|$ in $G[C \cup  H]$.
	Since $p > s$ and the elements in $\calV_H$ are pairwise disjoint, every element of $\calV_H$ contains at least one vertex of $S$. 
	Thus, we obtain $|S \cap (C \cup H)| \geq |H|$ and therefore $|S'| \leq |S| - |H|$.

For the other direction, assume  $(G',p-|H|)$ is a yes-instance and let  $S'$ be a solution for $(G',p-|H|)$, and let $\max_{Q \in \comp(G'-S')} |Q| =s'$. We claim that $S=(S'\cup H)\cap V$ is a solution for $(G,p)$. 

Components of $G-S$ are either the components of $G[C \cup H] - H$, since  $H$ separates $C$ from the rest of $G$, or components of $G'-S'$. 
By definition of the crown decomposition, all components of $G[C \cup H] - H$ are of size at most $c$, while all components of $G'-S'$ are of size at most $s'$ by the definition of $s'$. If $c\leq s'$ it follows that all components of $G-S$  have size at most $s'\leq p-|H|-|S'|\leq p-|S|$ (since $S'$ is a solution for $(G',p-|H|)$ creating components of size $s'$, meaning that $s'+|S'|\leq p-|H|$, and $|S|\leq |S'|+|H|$), and we are done. 

Assume that $s'< c$. This means that the $K_c$ component in $G'$ is shrunk to have size at most $s'$, meaning that $S'$ contains at least $c-s'$ vertices from $K_c$. This gives $|S|\leq |S'|+|H|-c+s'$.
If $G - S$ now contains a component larger than $s'$, this component originates from the crown and has a size of at most $c$. This ensures that $S$ is a solution for $(G, p)$ because 


	\begin{align*}
		|S|+c&\leq |S'|+|H|-c+s'+c\\
		&= |S'|+s'+|H| \\
		& \leq p-|H|+|H|\\
		&\leq p.	
	\end{align*}
\end{proof}

We will eventually use BCDs to find a good estimate for $c$ and to find a CD that can be removed with Lemma~\ref{lemma::saveRedVerInt}. Since we only search for $\lambda$-BCDs in components of size larger than $\lambda$, we still need another reduction rule for small components. Clearly, if an instance contains more than $p$ components of size at most $c$, deleting a smallest such component is a safe rule. With the same trick to add a $K_c$ like in Lemma~\ref{lemma::saveRedVerInt} we can go further and delete all other small components.

\begin{lemma}
	\label{lemma::small_components}
Let $(G,p)$ be an instance of \normalfont{VI}, $c<\pl$, and let $U$ be a connected component of $G$ of size at most $c$.  Then  $(G,p)$ is a yes-instance if and only if $(G-U+K_c,p)$ is a yes-instance.
\end{lemma}
\begin{proof}
If $(G, p)$ is a yes-instance, let $S$ be a corresponding solution, meaning that for $s = \max_{Q \in \comp(G - S)} |Q|$, we have $|S| + s \leq p$. By the definition of $\pl$, we know that $s \geq \pl$. Each component of $G - U + K_c - S$ is a component of $G - S$, and thus has size at most $s$, and the $K_c$, which has size $c \leq \pl \leq s$. Therefore, $S$ is also a solution for $(G - U + K_c, p)$.	

For the other direction, assume $(G-U+K_c,p)$ is a yes-instance and let $S$ be such that  $s=\max_{Q \in \comp(G-U+K_c-S)} |Q|$ satisfies $|S|+s\leq p$. We claim that $S'=S\cap V(G)$ is a solution for $(G,p)$. Since $U$ is a connected component of $G$, all components in $G-S'$, except for  $U$, are components of $G-U+K_c-S$, so they have size at most $p-|S|\leq p-|S'|$. Thus if  $|U|\leq s$ then it immediately follows that all components of $G-S'$ have size at most $p-|S'|$ meaning that $S'$ is a solution for $(G,p)$.

If $|U|> s$ then $U$ is a largest component in $G-S'$. Since $|U|\leq c$, it follows that $c>s$ which means that the $K_c$ component in $G-U+K_c$ must be shrunk by $S$. More precisely, $S$ contains at least $c-s$ of its vertices. This means $|S'|\leq |S|-c+s$ which gives $|U|+|S'|\leq c+|S|-c+s=|S|+s\leq p$. With $U$ being a largest component in $G-S'$ it again follows that $S'$ is a solution for $(G,p)$.
\end{proof}

Note that exhaustive application of  Lemma~\ref{lemma::small_components} creates an equivalent instance with only one component of size at most $c$ (isomorphic to $K_c$).

In order to use our two reduction rules, we must first determine a suitable lower bound for $\pl$. 
Additionally, we need to ensure the existence of a corresponding weighted crown decomposition when the input graph size exceeds a certain size, and that we can  find it efficiently.
This is where the BCD comes into play.
For $\lambda \in \N$ a $\lambda$-BCD $\chrf$ in $G$ can only be computed within the components of $G$ that have a size (or weight) greater than $\lambda$.
As a reminder, when we compute a $\lambda$-BCD, we disregard small components, meaning that if  $(C,H,\calR,f)$ is a $\lambda$-BCD for $G$, then $C\cup H\cup V(\calR)$ are only the vertices that are contained in some component of size more than $\lambda$ of $G$.

We specify the following lemma directly in the weighted version to also use it later.
The lemma provides a lower bound for $\pl$ from a $\lambda$-BCD $(C,H,\calR,f)$, where we essentially use that a $\lambda$-BCD is also a $(\lambda+1)$-packing of size $|H| + |\calR|$.

\begin{lemma}
	\label{lemma::lbVerInt}
	Let $(G,p)$ be an instance of \normalfont{wVI} and for $\lambda \in [p]$ let $(C,H,\calR,f)$ be a $\lambda$-\normalfont{BCD} in $G$.
	If $|H| + |\calR| > p$, then $\lambda < \pl$ for the instance $(G[C \cup H \cup V(\calR)],p)$.
\end{lemma}
\begin{proof}
	Let $V'=C \cup H \cup \calR'$ and consider the graph $G[V']$ (subgraph of $G$ created by removing all components of  weight at most $\lambda$).
	The vertex sets $\calR \cup \{\{h\} \cup V(f^{-1}(h))\}_{h \in H}$ form a $(\lambda+1)$-packing of size $|\calR| + |H| > p$ in  $G[V']$.
	It follows that we need more than $p$ vertices for a vertex separator $S$ in $G[V']$, such that $\max_{Q \in \comp(G-S)} w(Q) \leq \lambda$. Thus for any $U \subset V$ with $|U| \leq p$ it follows that $\max_{Q \in \comp(G-U)} w(Q) > \lambda$.
	This implies $\lambda < \pl$ for the instance $(G[V'],p)$, since any solution has cardinality at most $p$.
\end{proof}

An additional advantage of a $(\lambda+1)$-BCD $(C,H,\calR,f)$ is that the balanced part $\calR$ can be upper bounded by $3(\lambda+1)|\calR| \leq 3\pl|\calR| < 3p|\calR|$, while a suitable $\lambda$ choice also upper bounds $|\calR|$ by $p$.
In particular, if $H=C=\varnothing$ and $|\calR| \leq p$, then we would already have an instance with a size (or weight) of at most $3p^2$.

Clearly, a yes-instance cannot have a lower bound for $\pl$ larger than $p$ as stated in the following corollary. 

\begin{corollary}
	\label{corollary::noInstVerInt}
	Let $(G,p)$ be an instance of \normalfont{wVI}.
	If for a $p$-\normalfont{BCD} $(C,H,\calR,f)$ it holds that $|H| + |\calR| > p$, then $(G,p)$ is a no-instance.
\end{corollary}

The next lemma shows under which conditions we can find a $(p,\lambda)$-CD in $G$ with respect to a $\lambda$-BCD and the current graph size.

\begin{lemma}
	\label{lemma::CDexists}
	Let $(G,p)$ be an instance of \normalfont{VI} and for $\lambda \in [p]$ let $(C,H,\calR,f)$ be a $\lambda$-\normalfont{BCD}  in $G$ with $|H| + |\calR| \leq p$.
	If $|C| + |H| + |V(\calR)|  \geq p\max\{3\lambda,p+\lambda\}$, then $G'=G[C \cup H]$ contains a 
	$(p,\lambda)$-\normalfont{CD} in $G$.
	Furthermore, we can extract it from $G'$ in time $\O(|V(G')||E(G')|) \subseteq \O(|V(G)||E(G)|)$. 
\end{lemma}
\begin{proof}
Consider Theorem~\ref{thm::dbExapnsion} on $G'=G[H\cup C]$ (so $H$ in the role of $A$ and $C$ in the role of $B$), and with $d_h=p-1+\lambda$ and $y=\lambda$. This theorem then gives us a DBE $(H_1,H_2,\lambda,f,p-1+\lambda)$, with the guarantee $H_1\neq \emptyset$ in particular if $|H|+|C|\geq \sum_{h\in H} d_h$ (case $A=A'$ in the Theorem). The part $(C_1,H_1,f)$ with $C_1=f^{-1}(H_1)$ then gives us the desired non-empty $(H_1,\emptyset,\lambda,f,p-1+\lambda)$ DBE, in other words a $(p,\lambda)$-CD in $G'$.
	Moreover, observe that then $N(C) \subseteq H$ in $G$ and $H_1 \subseteq H$, with $N(C_1) \subseteq H_1$ in $G'$ implies $N(C_1) \subseteq H_1$ in $G$ and hence, $(C_1,H_1,f)$ is also a $(p,\lambda)$-CD in $G$.  
	
Thus it remains to show $|H|+|C|\geq \sum_{h\in H} d_h= (p-1+\lambda)|H|$. Suppose $|C| + |H| < (p-1+\lambda)|H|$, then:

	\begin{align*}
		|R| + |C| + |H| &< 3\lambda|\calR| + (p-1+\lambda)|H|\\
		&\leq 3\lambda|\calR| + (p+\lambda)|H|\\
		&\leq \max\{3\lambda,p+\lambda\}(|\calR| + |H|)\\ 
		&\leq \max\{3\lambda,p+\lambda\}p,
	\end{align*}
	which is a contradiction to $|R| + |C| + |H| \geq \max\{3\lambda,p+\lambda\}p$.
	
	It follows that $|C| \geq (p-2+\lambda)|H|$ and thus $|H|+|C|\geq |H|+(p-2+\lambda)|H|= \sum_{h\in H} d_h$.  We can thus use Theorem~\ref{thm::dbExapnsion} to extract a $(p,\lambda)$-CD $(C_1,H_1,f)$ with $\emptyset \neq H_1 \subseteq H$ and $C_1 \subseteq C$ from $G'$ in time $\O(|V(G')||E(G')|)$. 
\end{proof}

If we now combine Lemma~\ref{lemma::lbVerInt} and  Lemma~\ref{lemma::CDexists}, we can finally use  Lemma~\ref{lemma::saveRedVerInt} by finding a $\lambda \in [p-1]$ such that a ($\lambda+1$)-BCD $(C,H,\calR,f)$ satisfies $|H| + |\calR| \leq p$, while a $\lambda$-BCD $(C',H',\calR',f')$ satisfies $|H'| + |\calR'| > p$, which only adds a factor of $\O(\log(p))$ to the computation cost.
Formulated more precisely: For $Q \in \comp(C)$ we have $|Q| \leq \lambda+1 \leq p_\ell$ as $\lambda < \pl$ by  Lemma~\ref{lemma::lbVerInt}.
If the size of~$G[C \cup H \cup V(\calR)]$ is large, then we can find a ($p,\lambda$)-CD for $G$ in $G[C \cup H]$ in polynomial time by  Lemma~\ref{lemma::CDexists}, which is a reducible structure by  Lemma~\ref{lemma::saveRedVerInt}.

With these facts in hand we can design a kernelization algorithm.
It takes as input an instance $(G,p)$ of VI and returns an equivalent instance with at most $3p^2$ vertices.

\paragraph*{Find reducible structures (\algVI)}
\begin{enumerate}
	\item
	\label{step::verIntNoInst} 
	Compute a $p$-BCD $(C_0,H_0,\calR_0,f_0)$.
	If $|H_0| + |\calR_0| > p$, return a trivial no-instance.   
	\item
	\label{step::verInt1BCD}
	Compute a $1$-BCD $(C_1,H_1,\calR_1,f_1)$. If $|H_1| + |\calR_1| \leq p$, set $c=1$
	\item
	\label{step::verIntBinarySearch}
	Otherwise, for $\lambda \in [p-1]$ use binary search to find a ($\lambda+1$)-BCD $(C_1,H_1,\calR_1,f_1)$ and a $\lambda$-BCD $(C_2,H_2,\calR_2,f_2)$, such that $|H_1| + |\calR_1| \leq p$ and $|H_2| + |\calR_2| > p$. Set $c=\lambda+1$.
	\item\label{step::crownred} If $|C_1| + |H_2| + |V(\calR_2)|  \geq p\max\{3c,p+c\}$, compute a $(p,c)$-CD $(H',C',f')$ in $G[C_1 \cup H_1]$ and return $(G-(H'\cup C')+K_c,p-|H'|)$.
	\item\label{step::smallred} Delete all components of size at most $c$ and add a copy of $K_c$ in $G$. 
\end{enumerate}

\begin{lemma}
For any instance $(G,p)$ with $p\geq 2$ and where not all components of $G$ have size at most $p$, exhaustive application of \algVI produces an equivalent instance  with at most $3p^2$ vertices.
\end{lemma}
\begin{proof}
We first prove that an application of \algVI is safe. 
The correctness of step~\ref{step::verIntNoInst} is implied by Corollary~\ref{corollary::noInstVerInt}.
Assuming $c$ satisfies $c\leq \pl$, the correctness of steps~\ref{step::smallred} and~\ref{step::crownred} follow directly from exhaustive application of Lemma~\ref{lemma::small_components} and Lemma~\ref{lemma::saveRedVerInt}, respectively.

For correctness, it remains to prove that $c$ satisfies $c\leq \pl$. Since the input $G$ is a graph with at least one component of size more than $p$, the choice $c=1$ in step~\ref{step::verInt1BCD} certainly satisfies $\pl\geq c$.
For  step~\ref{step::verIntBinarySearch}  we first have to show that the claimed value for $\lambda$ with a  ($\lambda+1$)-BCD $(C_1,H_1,\calR_1,f_1)$ and a $\lambda$-BCD $(C_2,H_2,\calR_2,f_2)$, such that $|H_1| + |\calR_1| \leq p$ and $|H_2| + |\calR_2| > p$ exists. For a $\lambda$-BCD computation in a binary search for the correct value for $\lambda$, we cannot ensure that $|H| + |\calR|$  increases with decreasing $\lambda$ values, because the sizes of $H$ and $\calR$ are not necessarily monotonic with respect to $\lambda$. (Note, for example, that the elements in $\calR$ have a size range from $\lambda+1$ to $3\lambda$.) For a successful binary search, however, it is sufficient to know that there is a $\lambda$ that satisfies the desired properties for step~\ref{step::verIntBinarySearch}. Since we only enter this step if the extreme cases ($p$-BCD and 1-BCD) in step~\ref{step::verIntNoInst} and~\ref{step::verInt1BCD} hold, there has to exist at least one such $\lambda$ value in-between. By Lemma~\ref{lemma::lbVerInt}, a $\lambda$-BCD $(C_2,H_2,\calR_2,f_2)$ with $|H_2| + |\calR_2| > p$ implies $\pl>\lambda$. Since $\pl$ is an integer, it follows that $\pl \geq \lambda+1=c$.

For the size of the resulting instance, first note that a trivial no-instance has constant size. Otherwise, after step~\ref{step::smallred} the only component of size at most $c$ in $G$ is one copy of $K_c$. For the larger components, step~\ref{step::crownred} reduces further, if the $(p,c)$-CD is non-empty. By Lemma~\ref{lemma::CDexists}, a non-empty $(p,c)$-CD can be found if $|C_1| + |H_2| + |V(\calR_2)|  \geq p\max\{3c,p+c\}$. Thus exhaustive application of \algVI terminates with a $c$-BCD $(C_1,H_1,\calR_1,f_1)$ that satisfies  $|C_1| + |H_2| + |V(\calR_2)|  \leq p\max\{3c,p+c\}$ with a $c\leq\pl$. The union $C_1\cup H_2\cup V(\calR_2)$ contains all vertices of the components in $G$ that have size larger than $c$, while after step~\ref{step::smallred} only one component of size at most $c$ remains. This gives a total of at most $p\max\{3c,p+c\}+c$  vertices in $G$. With $c\leq \pl <p$, if follows that the reduced graph contains at most $p(3(p-1))+p-1<3p^2$ vertices.

\end{proof}

Finally, we prove the specified running time of Theorem~\ref{theorem::kernelVerInt}, which completes the proof of Theorem~\ref{theorem::kernelVerInt}.
Note that for a yes-instance we have to call the algorithm {\algVI} at most $|V|$ times to guarantee the desired kernel.

\begin{lemma}
	\label{lemma::runtimeVI}
	Algorithm \algVI runs in time $\O(\log(p)|V|^3|E|)$.
\end{lemma}
 
\begin{proof}
	The calculation of a $\lambda$-BCD for $\lambda \in [p]$ runs in time $\O(|V|^3|E|)$ by Theorem~\ref{theorem:lccd}, since a ($\lambda+1$)-packing has size less than $|V|$. 
	The extraction of a $(p,\lambda)$-CD from it runs in time $\O(|V||E|)$ by Theorem~\ref{thm::dbExapnsion}.
	Each step consists of the calculation of these constructions with an additional time expenditure of $\O(\log(p))$ in step~\ref{step::verIntBinarySearch}.
	As a result, the algorithm \algVI runs in time $\O(\log(p) |V|^3|E|)$.
\end{proof}

%% file: weighted_vertex_intergrity_problem.tex

In this section, we shift our focus to the weighted variants, namely wVI and wCOC.
While utilizing the packing size of the associated BCD offers a starting point for deriving a lower bound for VI, it proves insufficient for improvements in the weighted setting.
Therefore, we integrate the weight of the separator $H$ within a BCD $\chrf$ into our analysis. 
Additionally, we incorporate two distinct \dbe's into a BCD in case a reduction is not achievable within the respective setting.
This approach enables us to establish a tighter lower bound for $\pl$, estimate the remaining instance size more accurately to obtain the desired kernelization results (cf.~\cref{theorem::kernelWeightVerInt,theorem::wCOC}), or identify a no-instance.

We begin by introducing the two \dbe structures mentioned in conjunction with a BCD, which serve as preconditions for the upcoming lemmas.
Let $G=(V,E,w)$ be a vertex weighted graph, $a,s,\lambda \in \N$ and $\chrf$ a $\lambda$-BCD with $\lambda \in [a]$ and $|H| + |\calR| \leq s$, where $\max_{v \in V} w(v) \leq \max(a,s) =: \mu$.
Let $D^Y = \{d^Y_h\}_{h \in H}$ and $D^Z = \{d^Z_h\}_{h \in H}$ be demands, where $d^Y_h = a-2 + \lambda \cdot (w(h)+1)$ and $d^Z_h = w(h)-1 + (\sqrt{s}+1) \lambda$ for each $h \in H$.
Let $(Y_1,Y_2,\lambda,f_Y,D^Y)$ and $(Z_1,Z_2,\lambda,f_Z,D^Z)$ be \dbe's in $G[C \cup H]$ with $Y_1,Y_2$ as well as $Z_1,Z_2$ are partitions of $H$ and $f_Y^{-1}(Y_1),f_Y^{-1}(Y_2)$ as well as $f_Z^{-1}(Z_1),f_Z^{-1}(Z_2)$ are partitions of $\comp(G[C])$.

The concept of two \dbe's embedded in a BCD can be understood as follows: With an appropriate choice of parameters $a$, $s$, and $\lambda$, the substructures $Y_1$ and $f_Y^{-1}(Y_1)$ within the input graph of a wVI or wCOC instance form a reducible structure.
Note that if $Y_1 \neq \varnothing$ (or $Z_1 \neq \varnothing$), these sets act as separators, isolating $f_Y^{-1}(Y_1)$ (or $f_Z^{-1}(Z_1)$) from the rest of the graph, as $H$ separates $C$ from the rest. However, if $Y_1 = \varnothing$, attention shifts to $Z_1$. If the weight $w(Z_1)$ exceeds a certain threshold, we can analyze the heaviest remaining component after removing a separator with bounded weight.
If both $Y_1 = \varnothing$ and $w(Z_1)$ is bounded, we are able to bound the entire graph, or more precisely, all parts of the underlying $\lambda$-BCD $\chrf$.

Next, we present two crucial lemmas that serve as the foundation for designing the algorithms needed to prove \cref{theorem::kernelWeightVerInt,theorem::wCOC}. The first lemma aims to estimate the instance size.

\begin{lemma}
	\label{lemma::VIandWCOCEstimation}
	Let $w(Z_1) \leq 2 \mu^{1.5}$.
	If $Y_1 = \varnothing$, then $w(\calR) + w(C \cup H) < 3\mu(s + \sqrt{\mu}\lambda)$.
\end{lemma}
\begin{proof}
	Let $C_1 = V(f_Z^{-1}(Z_1))$ and $C_2 = V(f_Z^{-1}(Z_2))$.  
	We will use the demands $D^Y$ and $D^Z$ in the corresponding demanded balanced expansion to upper bound $w(Z_1 \cup C_1)$ and $w(Z_2 \cup C_2)$, respectively.  
	Note that $C = C_1 \cup C_2$ and $H = Z_1 \cup Z_2$.  
	Moreover, we use the properties of a $\lambda$-BCD to bound the weight of $R = \bigcup_{V' \in \calR} V'$ by $3 |\calR| \lambda$, since $w(V') \leq 3 \lambda$ for every $V' \in \calR$.
	
	First, observe that $f_Z^{-1}(Z_1) \subseteq f_Y^{-1}(Z_1)$ since the components $f_Z^{-1}(Z_1)$ are separated by $Z_1$ and must be assigned to $Z_1$ in $f_Y$.  
	Note that $Y_2 = H$ since $Y_1 = \varnothing$, which implies that $Z_1 \subseteq Y_2$.  
	These facts allow us to bound $w(Z_1 \cup C_1)$ by the demands $D_Y$.  
	As a result, we obtain

	\begin{align*}
		w(Z_1 \cup C_1) &= \sum_{h \in Z_1} \left(w(h) + w(f_Z^{-1}(h))\right)\\
		&\leq \sum_{h \in Z_1} \left(w(h) + w(f_Y^{-1}(h))\right)\\
		&\leq \sum_{h \in Z_1} \left(d^Y_h + \lambda -1\right)\\
		&\leq \sum_{h \in Z_1} \left(a-2 + \lambda \cdot (w(h)+1)+\lambda-1\right)\\
		&< \sum_{h \in Z_1} \left(a + 2\lambda + w(h) \cdot \lambda\right).
	\end{align*}
	Furthermore, using the demands $D_Z$ we can bound $w(Z_2 \cup C_2)$, since for each $h \in Z_2$ we have
	$$w(h) + w(f_Z^{-1}(h)) \leq d_h^Z + \lambda-1 \leq  w(h)-1 + (\sqrt{s}+1) \lambda + \lambda-1 < w(h) + 2\lambda + \sqrt{s} \lambda.$$ 
	Now a straight forward calculation provides the proof:
	\begin{align*}
	 w(\calR) + w(C \cup H) &\leq 3 |\calR| \lambda + w(C_1 \cup Z_1) + w(C_2 \cup Z_2)\\
		&< 3 |\calR| \lambda + \sum_{h \in Z_1} (a + 2\lambda + w(h) \cdot \lambda) + \sum_{h \in Z_2} (w(h) + 2\lambda + \sqrt{s} \lambda)\\
		&=3 |\calR| \lambda + |Z_1|(a+2\lambda) + |Z_2|(\sqrt{s}\lambda+2\lambda) + \lambda \cdot \sum_{h \in Z_1} w(h) + \sum_{h \in Z_2} w(h)\\
		&\leq 3 |\calR| \lambda + (|Z_1| + |Z_2|)2\lambda + |Z_1|a + |Z_2| \sqrt{s}\lambda + w(Z_1) \cdot \lambda  + w(Z_2)\\
		&< 3 |\calR| a + 2|H| a + |Z_1| a +  \mu^{1.5}\lambda + 2 \mu^{1.5} \lambda + |Z_2|\mu\\
		&\leq 3 |\calR| \mu + 2|H| \mu + (|Z_1| + |Z_2|)\mu + 3 \mu^{1.5} \lambda\\
		&\leq 3\mu(|\calR| + |H|) + 3 \mu^{1.5} \lambda\\ 
		&\leq 3\mu(s + \sqrt{\mu}\lambda)
	\end{align*}
\end{proof}

With respect to the weight of $Z_1$, the second lemma is relevant to \cref{lemma::VIandWCOCEstimation} and helps to establish a more precise lower bound for $\pl$, or to identify instances where no feasible solutions exist.

\begin{lemma}
	\label{lemma::noInstOrLowerBound}
	If $w(Z_1) > 2 \mu^{1.5}$, then there is no separator $S$ such that $w(S) \leq s$ and $\max_{Q \in \comp(G-S)} w(Q) \leq \lambda$.
\end{lemma}
\begin{proof}
	Let $Z_1^{>\sqrt{\mu}} := \{h \in Z_1 \mid w(h) > \sqrt{\mu}\}$ and $Z_1^{\leq \sqrt{\mu}} := Z_1 \setminus Z_1^{>\sqrt{\mu}}$.
	We have $\left|Z_1^{\leq \sqrt{\mu}}\right| \leq |H| \leq s \leq \mu$ and therefore $w\left(Z_1^{>\sqrt{\mu}}\right) > 2 \mu^{1.5} - w\left(Z_1^{\leq \sqrt{\mu}}\right) \geq 2 \mu^{1.5} - \mu \cdot \sqrt{\mu} = \mu^{1.5}$.
	It follows that $\left|Z_1^{>\sqrt{\mu}}\right| > \sqrt{\mu} \geq \sqrt{s}$ as $w(h) \leq \mu$ and $\mu \geq s$.
	
	Let $V_h = \{h\} \cup f_Z^{-1}(h)$ for each $h \in Z_1$.
	Fix an $h \in H$.
	We have $w(h) > \sqrt{\mu} \geq \sqrt{s}$ and $w(V_h) > w(h)-1 + (\sqrt{s}+1) \lambda - \lambda + 1 = w(h) + \sqrt{s} \lambda$, where $G[V_h]$ is connected.
	Furthermore, observe that if we do not remove $h$, we must remove at least $\sqrt{s}$ vertices to guarantee that $G[V_h]$ does not contain any component of weight larger than $\lambda$, since removing $h$ decomposes $G[V_h]$ into components with weight at most $\lambda$.
	In other words, for any $V' \subset V_h \setminus \{h\}$ with $|V'|<\sqrt{s}$ we have 
	$\max_{Q \in \comp(G[V_h] - V')} w(Q) \geq w(V_h) - \lambda |V'| > w(h) + \lambda \sqrt{s} - \lambda \cdot (\sqrt{s}-1) > \lambda$.
	That is, to ensure that $G[V_h]$ contains no component of weight larger than $\lambda$ we either need to remove $h$ from it, or remove at least $\sqrt{s}$ vertices of $V_h \setminus \{h\}$.
	Hence, both options cost us at least $\sqrt{s}$ in $G[V_h]$.
	Since the vertex sets $\{V_h\}_{h \in H}$ are pairwise disjoint and are more than $\sqrt{s}$, a separator $S$ ensuring that there is no component of weight larger than $\lambda$ cost us more than $\sqrt{s} \cdot \sqrt{s}=s$ in $G[Z_1 \cup f_Z^{-1}(Z_1)] \subseteq G$.
\end{proof}

With the key lemmas for the wVI and wCOC problems established, we now demonstrate their applications for each problem and design the corresponding algorithms to achieve the desired kernelizations.

\subsubsection*{Weighted Component Order Connectivity Problem}
We begin by considering wCOC and start with a reduction rule based on a specific \dbe.
We assume that each vertex weighs at most $W$, since any vertex with a larger weight must necessarily be part of the solution.
That is, it is a safe reduction rule for an instance $(G,k,W)$ of wCOC to remove such vertices and update $k$ accordingly.
Moreover, we assume that every connected component in an instance of wCOC weighs more than $W$. 
It is not difficult to see that components with lesser weight can be safely removed.

\begin{lemma}
	\label{lemma::saveRedwCOC}
	Let $(G,k,W)$ be an instance of the weighted component order connectivity problem, $(H,H',W,f,D)$ a \normalfont{\dbe} in $G$ with $D = \{d_h\}_{h \in H}$, where $d_h = W-2 + W \cdot (w(h)+1)$ for each $h \in H$ and for $C=V(f^{-1}(H))$ we have $N(C) \subseteq H$.
	Then, $(G,k,W)$ is a yes-instance if and only if $(G-(H \cup C),k-w(H),W)$ is a yes-instance.
\end{lemma}
\begin{proof}
	Let $S$ be a solution of $(G,k,W)$, $S' = S \setminus (H \cup C)$ and $G' = G- (H \cup C)$.
	We prove that $S'$ is a solution for $(G',k,W)$, implying that if $(G,k,W)$ is a yes-instance then $(G',k,W)$ is a yes-instance.	
	We need to show that $w(S') \leq w(S) - w(H)$ and $\max_{Q \in \comp(G'-S')} w(Q) \leq W$.
	For the former, let $V_h = \{h\} \cup f^{-1}(h)$ for each $h \in H$.
	Note that $w(V_h) > d_h - W + 1= W - 1 + W \cdot w(h)$, where each $G[V_h]$ is connected.
	Moreover, observe that if we do not remove $h$, we must remove at least $w(h)$ vertices to guarantee that $G[V_h]$ does not contain any component of weight larger than $W$, since removing $h$ decomposes $G[V_h]$ into components with weight at most $W$.
	In other words, for any $V' \subset V_h \setminus \{h\}$ with $|V'|<w(h)$ we have 
	$\max_{Q \in \comp(G[V_h] - V')} w(Q) \geq w(V_h) - W|V'| > W-1 + W \cdot w(h) - W \cdot (w(h)-1) \geq W$ as $W \geq 1$. 
	That is, to ensure that $G[V_h]$ contains no component of weight $W+1$ we either need to remove $h$ from it, or remove at least $w(h)$ vertices of $V_h \setminus \{h\}$.
	Since the vertex sets $\{V_h\}_{h \in H}$ are pairwise disjoint, we have $w(S \cap (C \cup H)) \geq \sum_{h \in H} w(h) = w(H)$ and therefore $w(S') \leq w(S) - w(H)$.
	For the latter, suppose that $G'-S'$ contains a connected component $G''$ of weight larger than $W$.
	Note that $V(G'') \subset V \setminus (H \cup C)$.
	Since we do not remove vertices from $S$ in $V \setminus (C \cup H)$, we have $G'' \subseteq G-S$, contradicting that $S$ is a solution for $(G,k,W)$.
	
	The other direction follows simply because the components of $G[C \cup H] - H$ are of weight at most $W$ while $H$ separates $C$ from the rest of the graph, i.e., $N(C) \subseteq H$.
	That is, adding $H$ to a solution $S''$ in $G'$, which costs us additional $w(H)$, ensures that the heaviest component of $G-(S'' \cup H)$ weighs at most $W$.
\end{proof}

Derived from \cref{lemma::VIandWCOCEstimation,lemma::noInstOrLowerBound} we specify the corresponding lemmas with the appropriate parameter selection as corollaries to design the kernelization algorithm.
We define the BCD and \dbe structures properly for wCOC, but to facilitate a direct comparison, the specific parameter choices regarding \cref{lemma::VIandWCOCEstimation,lemma::noInstOrLowerBound} would be $a = W$, $s = k$, and $\lambda = W$.
Note that $a=\lambda$ satisfies the precondition $\lambda \in [a]$ and $\max_{v \in V(G)} w(v) \leq W$ for an instance $(G,k,W)$ of wCOC the precondition $\max_{v \in V(G)} w(v) \leq \max(a,s)$ of \cref{lemma::VIandWCOCEstimation,lemma::noInstOrLowerBound}.

Let $(G,k,W)$ be an instance of wCOC, $\chrf$ a $W$-BCD in $G$ with $|H| + |\calR| \leq k$ and $\mu = \max(k,W)$.
Let $D^Y = \{d^Y_h\}_{h \in H}$ and $D^Z = \{d^Z_h\}_{h \in H}$ be demands, where $d^Y_h = W-2 + W \cdot (w(h)+1)$ and $d^Z_h = w(h)-1 + (\sqrt{k}+1) W$ for each $h \in H$.
Let $(Y_1,Y_2,W,f_Y,D^Y)$ and $(Z_1,Z_2,W,f_Z,D^Z)$ be \dbe's in $G[C \cup H]$ 
with $Y_1,Y_2$ as well as $Z_1,Z_2$ are partitions of $H$ and $f_Y^{-1}(Y_1),f_Y^{-1}(Y_2)$ as well as $f_Z^{-1}(Z_1),f_Z^{-1}(Z_2)$ are partitions of $\comp(G[C])$.
Note that $(Y_1,Y_2,W,f_Y,D^Y)$ coincide with the \dbe in \cref{lemma::saveRedwVI}.
As a reminder, we assumed that only vertices in components of $G$ with weight greater than $W$ are part of the $W$-BCD. Since we have already established that $G$ contains only components with weight greater than $W$, it follows that all vertices of $G$ are part of a $W$-BCD in $G$.

\begin{corollary}
	\label{corollary::RedStructExWCOC}
	Let $(G,k,W)$ be an instance of wCOC and let $w(Z_1) \leq 2 \mu^{1.5}$.
	If $w(V(G)) \geq 3\mu(k + \sqrt{\mu}W)$, then $Y_1 \neq \varnothing$.
\end{corollary}

\begin{corollary}
	\label{corollary::NoInstWCOC}
	Let $(G,k,W)$ be an instance of wCOC.
	If $w(Z_1) > 2 \mu^{1.5}$, then $(G,k,W)$ is a no-instance.
\end{corollary}

The algorithm takes as input an instance $(G, k, W)$ of wCOC, where $w(G) \geq 3\mu(k + \sqrt{\mu}W)$.
It either outputs a reduced instance or determines that $(G, k, W)$ is a no-instance by returning a trivial no-instance.
Note that the weight condition on $w(G)$ is not a loss of generality, as there is nothing to do if $w(G) < 3\mu(k + \sqrt{\mu}W)$.

\paragraph*{Find reducible structures (\algWCOC):}
	
\begin{enumerate}
	\item Compute a $W$-BCD $(C,H,\calR,f)$ in $G$. If $|H| + |\calR| > k$, then return a trivial no-instance.
	\item Let $d_h = W-2 + W \cdot (w(h)+1)$ for each $h \in H$ and $D = \{d_h\}_{h \in H}$.
	Compute a \dbe $(H_1,H_2,W,f,D)$ in $G[C \cup H]$ with $H_1,H_2 \subseteq H$ and $f^{-1}(H_1),f^{-1}(H_2) \subseteq C$.
	\item If $H_1 = \varnothing$, then return a trivial no-instance.
	Otherwise, return $(G - (H_1 \cup f^{-1}(H_1)),k-w(H_1),W)$.
\end{enumerate}

With the following lemma we prove \cref{theorem::wCOC}.
Note that we need to apply \algWCOC at most $k$ times. 

\begin{lemma}
	Algorithm \algWCOC works correctly and has a runtime of $\O(r^2|V||E|)$, where $r$ is the size of a maximum $(W+1)$-packing.
\end{lemma}
\begin{proof}
	If $|H| + |\calR| > k$, then the algorithm returns a trivial no-instance.
	This is correct as $\calR \cup \{\{h\} \cup V(f^{-1})(h)\}_{h \in H}$ forms a $(W+1)$-packing in $G$ of size larger than $k$, which in turn shows that we need to remove more than $k$ vertices to achieve that the remaining graph contains only components of weight $W$.
	
	If $H_1 =  \varnothing$, then the case in \cref{corollary::NoInstWCOC} must be true, as otherwise, by the weight of $G$ and \cref{corollary::RedStructExWCOC} we must have $H_1 \neq \varnothing$ (in the corollary $Y_1 \neq \varnothing$).
	This shows that the conclusion that $(G,k,W)$ is a no-instance for the case $H_1 =  \varnothing$ is also correct.
	
	Otherwise, the algorithm reduces the instance $(G,k,W)$ to $(G - (H_1 \cup f^{-1}(H_1)),k-w(H_1),W)$, where the reduction is safe by \cref{lemma::saveRedwCOC}.
	
	Regarding the running time, computing a BCD or a \dbe as well as all other steps take not more time than $\O(r^2|V||E|)$ (cf.~\cref{thm::dbExapnsion,theorem:lccd}).
\end{proof}

\subsubsection*{Weighted Vertex Integrity Problem}

We now turn our attention to wVI and prove \cref{theorem::kernelWeightVerInt}.
As a reminder, we assume for an instance $(G,p)$ of wVI that there is a component of weight greater than $p$ and that every vertex has weight lesser than $p$; otherwise, it can be immediately labeled as a trivial yes-instance or no-instance, respectively.
Again, we begin with a reduction rule based on a specific \dbe.
We point out that one direction of the reduction is very similar to the reduction in \cref{lemma::saveRedwCOC} for wCOC.
However, the other direction requires more careful consideration.

\begin{lemma}
	\label{lemma::saveRedwVI}
	Let $(G,p)$ be an instance of the weighted vertex integrity problem, $\lambda \in \N$, where $1 \leq \lambda \leq \pl$, and $(H,H',\lambda,f,D)$ a \normalfont{\dbe} in $G$ with $D = \{d_h\}_{h \in H}$, where $d_h = p-2 + \lambda \cdot (w(h)+1)$ for each $h \in H$ and for $C=V(f^{-1}(H))$ we have $N(C) \subseteq H$.
	Then, $(G,p)$ is a yes-instance if and only if $(G- (H \cup C) + v,p-w(H))$ is a yes-instance, where $v$ is a vertex of weight $\lambda$.
\end{lemma}
\begin{proof}
	Let $S$ be a solution of $(G,p)$ with $c = \max_{Q \in \comp(G - S)} w(Q)$, $S' = S \setminus (H \cup C)$ and $G' = G- (H \cup C) + v$.
	We prove that $S'$ is a solution for $(G',p-w(H))$ with $\max_{Q \in \comp(G'-S')} w(Q) \leq c$, 
	implying that if $(G,p)$ is a yes-instance then $(G',p-w(H))$ is a yes-instance.
	We need to show that $w(S') \leq w(S) - w(H)$ and $\max_{Q \in \comp(G'-S')} w(Q) \leq c$.
	For the former, let $V_h = \{h\} \cup f^{-1}(h)$ for each $h \in H$.
	Note that $w(V_h) > d_h - \lambda + 1= p - 1 + \lambda \cdot w(h)$, where each $G[V_h]$ is connected.
	Moreover, observe that if we do not remove $h$, we must remove at least $w(h)$ vertices to guarantee that $G[V_h]$ does not contain any component of weight larger than $p$, since removing $h$ decomposes $G[V_h]$ into components with weight at most $\lambda$.
	In other words, for any $V' \subset V_h \setminus \{h\}$ with $|V'|<w(h)$ we have 
	$\max_{Q \in \comp(G[V_h] - V')} w(Q) \geq w(V_h) - \lambda|V'| > p-1 + \lambda \cdot w(h) - \lambda \cdot (w(h)-1) \geq p$ as $\lambda \geq 1$. 
	That is, to ensure that $G[V_h]$ contains no component of weight $p+1>c$ we either need to remove $h$ from it, or remove at least $w(h)$ vertices of $V_h \setminus \{h\}$.
	Since the vertex sets $\{V_h\}_{h \in H}$ are pairwise disjoint, we have $w(S \cap (C \cup H)) \geq \sum_{h \in H} w(h) = w(H)$ and therefore $w(S') \leq w(S) - w(H)$.
	For the latter, suppose that $G'-S'$ contains a connected component $G''$ of weight larger than $c$.
	Note that $V(G'') \subset V \setminus (H \cup C)$.
	Moreover, note that the added vertex $v$ in $G'$ is isolated and of weight $\lambda \leq \pl \leq c$ and thus, $v \notin V(G'')$.
	Since we do not remove vertices from $S$ in $V \setminus (C \cup H)$, we have $G'' \subseteq G-S$, contradicting $c = \max_{Q \in \comp(G-S)} w(Q)$.
	
	For the other direction, we first show that if $(G', p - w(H))$ is a yes-instance, then there is always a solution $R$ such that $\max_{Q \in \comp(G' - R)} w(Q) \geq \lambda$.
	Once we have this, the other direction follows simply because the components of $G[C \cup H] - H$ have weight at most $\lambda$, while $H$ separates $C$ from the rest of the graph, i.e., $N(C) \subseteq H$.
	Adding $H$ to $S''$ incurs an additional cost of $w(H)$, which shows that $R \cup H$ is a solution for $(G, p)$.

	Let $R'$ be a solution of $(G', p - w(H))$ with $c' = \max_{Q \in \comp(G' - R')} w(Q)$.  
	If $c' \geq \lambda$, then there is nothing to do, and we are done.  
	Otherwise, $v \in R'$, since $v$ is isolated in $G'$ and has weight $\lambda$.  
	We remove $v$ from $R'$, and let the new set be $R''$.  
	Note that $(\{v\}, \varnothing)$ is now the heaviest connected component in $G' - R''$, and therefore  
	$\max_{Q \in \comp(G' - R'')} w(Q) = w(v) = \lambda$.
	Finally, we need to show that $R''$ is a feasible solution in $G'$, which is demonstrated by the following computation:  
	$$p - w(H) \geq w(R') + c' = w(R') - \lambda + \lambda + c' \geq w(R') - w(v) + \lambda = w(R'') + \lambda.$$
\end{proof}

To design the upcoming algorithm, we need to establish an additional reduction rule, ensuring that isolated components with a weight of at most $\pl$ can be safely removed.

\begin{lemma}
	\label{lemma::saveRedSmallComp}
	Let $(G,p)$ be an instance of wVI that contains at least one component of weight greater than $p$, where $p>1$.
	Let $c \leq \pl$ and let $U \subset V(G)$ be all vertices within connected components of weight at most $c$.
	\begin{enumerate}
		\item If $c>1$, then $(G,p)$ is a yes-instance if and only if $(G-U+v, p)$ is a yes-instance, where $w(v) = c$.
		\item If $c=1$, then $(G,p)$ is a yes-instance if and only if $(G-U, p)$ is a yes-instance.
	\end{enumerate}
\end{lemma}
\begin{proof}
	\text{}
	\begin{enumerate}
		\item
		Let $G' = G - U + v$.
		Let $S$ be a solution of $(G, p)$, and define $x = \max_{Q \in \comp(V(G) - S)} w(Q)$.
		For $S' = S \setminus U$, it follows that $x \geq \max_{Q \in \comp(V(G) - S')} w(Q)$, since the largest component of $\comp(G[U])$ has weight $c \leq \pl \leq x$.
		As a result, $S'$ is a solution for $(G', p)$ because $G - U = G' - \{v\}$ and $w(v) = c \leq x$.

		Let $R$ be a solution of $(G', p)$, let $y = \max_{Q \in \comp(G' - R)} w(Q)$, and set $R' = R \setminus \{v\}$.
		First, note that $\max_{Q \in \comp(G' - R')} w(Q) \leq y$ since $v$ is an isolated vertex in $G'$.
		If $y \geq c$, then $R'$ is a solution for $(G, p)$, since $G - U = G' - \{v\}$ and $\max_{Q \in \comp(G[U])} w(Q) \leq c \leq y$.
		Otherwise, $v \in R$ as $y < c$.
		Now, note that $(\{v\}, \varnothing)$ is the heaviest connected component in $G' - R'$.
		It follows that $c = w(v) = \max_{Q \in \comp(G' - R')} w(Q) \geq \max_{Q \in \comp(G - R')} w(Q)$,
		since $\max_{Q \in \comp(G - U - R')} w(Q) \leq y < c$ as $G - U = G' - \{v\}$ and $\max_{Q \in \comp(G[U])} w(Q) \leq c$.
		Finally, we show that $R'$ is a feasible solution for $G$, proving the lemma:
		$$ p \geq w(R) + y = w(R) + y + c - c \geq w(R) - c + c = w(R) - w(v) + c = w(R') + c $$
		
		\item
		Let $S$ be a solution of $(G, p)$ and let $G' = G - U$.  
		Since there is a component with weight greater than $p$, we have $\max_{Q \in \comp(G - S)} w(Q) \geq 1$.  
		It follows that $S' = S \setminus U$ is also a solution for $(G, p)$, since $w(u) = 1$ for the isolated vertices $u \in U$.  
		Finally, $S'$ is also a solution for $(G', p)$ as $G' \subseteq G$.

		Let $R$ be a solution of $(G', p)$.  
		The graph $G$ contains a component $T$ with weight greater than $p \geq 1$, so $V(T) \cap U = \varnothing$.  
		Since $T$ is also in $G'$, we have $\max_{Q \in V(G')} w(Q) \geq 1$ because not all vertices of $V(T)$ can be in $R$.  
		As a result, $R$ is also a solution for $(G, p)$, since $w(u) = 1$ for the isolated vertices $u \in U$.
	\end{enumerate}

\end{proof}

%
%
%

Derived from \cref{lemma::VIandWCOCEstimation,lemma::noInstOrLowerBound} we specify the corresponding lemmas with the appropriate parameter selection as corollaries to design the kernelization algorithm.
We define the BCD and \dbe structures properly for wVI, but to facilitate a direct comparison, the specific parameter choices regarding \cref{lemma::VIandWCOCEstimation} and \cref{lemma::noInstOrLowerBound} would be $a = s = p$, and $\lambda$ is not fixed, as we do not know $\pl$. However, as in the preconditions of \cref{lemma::VIandWCOCEstimation,lemma::noInstOrLowerBound}, it is assumed that $\lambda \in [a]$.
Note also that $\max_{v \in V(G)} w(v) \leq p$ for an instance $(G,p)$ of wVI satisfies the precondition $\max_{v \in V(G)} w(v) \leq \max(a,s)$.

Let $(G,p)$ be an instance of wVI and $\chrf$ a $\lambda$-BCD in $G$  with $|H| + |\calR| \leq p$.
Let $D^Y = \{d^Y_h\}_{h \in H}$ and $D^Z = \{d^Z_h\}_{h \in H}$ be demands, where $d^Y_h = p-2 + \lambda \cdot (w(h)+1)$ and $d^Z_h = w(h)-1 + (\sqrt{p}+1) \lambda$ for each $h \in H$.
Let $(Y_1,Y_2,\lambda,f_Y,D^Y)$ and $(Z_1,Z_2,\lambda,f_Z,D^Z)$ be \dbe's in $G[C \cup H]$ with $Y_1,Y_2$ as well as $Z_1,Z_2$ are partitions of $H$ and $f_Y^{-1}(Y_1),f_Y^{-1}(Y_2)$ as well as $f_Z^{-1}(Z_1),f_Z^{-1}(Z_2)$ are partitions of $\comp(G[C])$.
Note that $(Y_1, Y_2, \lambda, f_Y, D^Y)$ coincides with the \dbe in \cref{lemma::saveRedwVI} only if $\lambda \leq \pl$.

\begin{corollary}
	\label{cor::wVI1}
	Let $(G,p)$ be an instance of wVI and let $w(Z_1) \leq 2 p^{1.5}$.
	If $w(\calR) + w(C) + w(H) \geq 3p(p + \sqrt{p}\lambda)$, then $Y_1 \neq \varnothing$.
\end{corollary}

\begin{corollary}
	\label{cor::wVI2}
	Let $(G,p)$ be an instance of wVI and let $w(Z_1) > 2 p^{1.5}$.
	\begin{enumerate}
		\item If $\lambda < p$, then $\lambda < \pl$.
		\item If $\lambda \geq p$, then $(G,p)$ is a no-instance.
	\end{enumerate}
\end{corollary}

With these corollaries and \cref{lemma::lbVerInt}, we can design the following kernelization algorithm.
It takes as input an instance $(G, p)$ of wVI, where $|V(G)| > 3p(p + \sqrt{p})$ and there is a connected component in $G$ with weight greater than $p$.
It either determines that $(G, p)$ is a no-instance by returning a trivial no-instance, or that $(G, p)$ already satisfies the desired kernel, or it outputs a reduced instance.
Note that the target vertex kernel contains a factor $\pl < p$ that we do not know beforehand, which is why an instance may already satisfy the vertex kernel. Moreover, a vertex kernel of $3p(p + \sqrt{p})$ is the smallest size we can guarantee from \cref{theorem::kernelWeightVerInt}, as $\pl \geq 1$.
Finally, if there is no component with weight greater than $p$, there is nothing to do. These facts show that the weight assumptions on $G$ are not a loss of generality.

To enhance the readability of the algorithm, we define $\lambda$-crown structures for a $\lambda \in \mathbb{N}$ and demands $D$ to denote a $\lambda$-BCD $\chrf$ and a \dbe $(H_1, H_2, \lambda, f, D)$ in $G[C \cup H]$, where $H_1, H_2$ forms a partition of $H$ and $f^{-1}(H_1), f^{-1}(H_2)$ a partition of $\comp(G[C])$.
As a reminder, when we compute a $\lambda$-BCD, we disregard small components, meaning that if  $(C,H,\calR,f)$ is a $\lambda$-BCD for $G$, then $C\cup H\cup V(\calR)$ are only the vertices that are contained in some component of weight more than $\lambda$ of $G$.

\paragraph*{Find reducible structures (\algWVI):}
\begin{enumerate}
	\item
	\label{step::wVI1}
	Let $(C^1,H^1,\calR^1,f_1)$ and $(H^1_1,H^1_2,1,g_1,\{d^1_h\}_{h \in H^1})$ be $1$-crown-structures with $d^1_h = w(h) + \sqrt{p}$.
	\begin{enumerate}
		\item
		\label{step::wVI1_a}
		If $U = V(G)\setminus (C^1 \cup H^1 \cup V(\calR^1))$ is not empty, then return $(G-U,p)$.
		\item 
		\label{step::wVI1_b}
		If $w(H^1_1) \leq 2 p^{1.5}$ and $|H^1| + |\calR^1| \leq p$,
		then compute a \dbe $(H^2_1,H^2_2,1,g_2,\{d^2_h\}_{h \in H^1})$ with $d^2_h = p-2 + w(h)+1$ and
		return $(G - (H^2_1 \cup g_2^{-1}(H^2_1)),p-w(H^2_1))$.  
	\end{enumerate}
	\item
	\label{step::wVI2}
	Otherwise, find for $\lambda \in [p-1]$ and for each $h \in H$ with demands 
	\begin{itemize}
		\item $d^3_h = w(h)-1 + (\sqrt{p}+1) \lambda$,
		\item $d^4_h = w(h)-1 + (\sqrt{p}+1)(\lambda+1)$
	\end{itemize}
	a $\lambda$-crown-structure and a ($\lambda+1$)-crown-structure 
	\begin{itemize}
		\item $(C^3,H^3,\calR^3,f_3)$, $(H^3_1,H^3_2,\lambda,g_3,\{d^3_h\}_{h \in H})$,
		\item $(C^4,H^4,\calR^4,f_4)$, $(H^4_1,H^4_2,\lambda + 1,g_4,\{d^4_h\}_{h \in H})$,
	\end{itemize}
	such that
	\begin{itemize}
		\item $w(H^3_1) > 2p^{1.5}$ and/or $|H^3| + |\calR^3| > p$,
		\item $w(H^4_1) \leq 2p^{1.5}$ and $|H^4| + |\calR^4| \leq p$.
	\end{itemize}
	\begin{enumerate}
		\item 
		\label{step::wVI2_a}
		If such a $\lambda$ does not exist, then return a trivial no-instance.
		\item 
		\label{step::wVI2_b}
		If $U' = V(G)\setminus (C^4 \cup H^4 \cup V(\calR^4))$ has cardinality larger than one, then return $(G-V''+\{v\},p)$, where $v$ is a vertex of weight $\lambda+1$.
		\item
		\label{step::wVI2_c}
		If $w(V(\calR^4)) + w(C^4) + w(H^4) < 3p(p + \sqrt{p} (\lambda+1))$, then report that  $|V(G)| \leq 3p(p + \sqrt{p} \pl)$.
	\end{enumerate}
	\item
	\label{step::wVI3}
	Otherwise, compute a \dbe $(H^5_1,H^5_2,\lambda + 1,g_5,\{d^5_h\}_{h \in H^4})$ with $d^5_h = p-2 + (\lambda+1) \cdot (w(h)+1)$ in $G[C^4 \cup H^4]$ with $H^5_1,H^5_2 \subseteq H^4$ and $g_5^{-1}(H^5_1),g_5^{-1}(H^5_2) \subseteq C^4$ and
	return $(G - (H^2_1 \cup g_5^{-1}(H^5_1) + v),p-w(H^2_1))$, where $v$ is a vertex of weight $\lambda+1$.
\end{enumerate}

With the following lemma, we prove \cref{theorem::kernelWeightVerInt}.  
Note that we need to apply \algWVI at most $|V|$ times until it either reports that the input instance already has the desired kernel or, after a reduction, the remaining instance is of size at most $3p(p + \sqrt{p})$.  
This is because, if the input graph does not have the desired kernel size or is labeled as a no-instance, each application of \algWVI reduces the input instance by at least one vertex.

\begin{lemma}
	Algorithm \algWVI works correctly and has a runtime of $\O(\log(p)|V|^3|E|)$.
\end{lemma}
\begin{proof}
	To prove that step~\ref{step::wVI1} works correctly, note that $\pl \geq 1$ for the instance $(G, p)$, since there is a connected component in $G$ with weight greater than $p$.  
	Step~\ref{step::wVI1_a} is correct by \cref{lemma::saveRedSmallComp}.  
	If in step~\ref{step::wVI1_b} we have $w(H^1_1) \leq 2p^{1.5}$ and $|H^1| + |\calR^1| \leq p$, then by \cref{cor::wVI1} and the fact that $w(G) > 3p(p + \sqrt{p})$, there must exist a corresponding \dbe\ with $H^2_1 \neq \varnothing$.  
	Additionally, the reduction executed is safe by \cref{lemma::saveRedwVI}, since $\pl \geq 1$.
	
	If in step~\ref{step::wVI2} the corresponding $\lambda$ does not exist, then by \cref{corollary::noInstVerInt} or \cref{cor::wVI2}, $(G, p)$ must be a no-instance, proving that step~\ref{step::wVI2_a} works correctly.
	Note that after step~\ref{step::wVI1} we have $w(H^1_1) > 2p^{1.5}$ and/or $|H^1| + |\calR^1| > p$.
	That is, there is at least for one $\lambda \in [p-1]$ a $\lambda$-BCD $\chrf$ with $w(H) > 2p^{1.5}$ and/or $|H| + |\calR| > p$ and if these properties also hold for a ($\lambda+1=p$)-BCD, then by \cref{corollary::noInstVerInt} or \cref{cor::wVI2} we have a no-instance as input.   
	Otherwise, a desired $\lambda$ exists and by \cref{lemma::lbVerInt} or \cref{cor::wVI2}, we have $\lambda + 1 \leq \pl$.  
	All vertices within connected components of weight at least $\lambda + 2$ are in $V\left(\calR^4\right) \cup C^4 \cup H^4$.  
	The remaining ones, i.e., the vertices within connected components of weight at most $\lambda + 1 \leq \pl$, are in $U'$.  
	Thus, by \cref{lemma::saveRedSmallComp}, step~\ref{step::wVI2_b} works correctly.  
	Finally, for step~\ref{step::wVI2}, we show that step~\ref{step::wVI2_c} also works correctly.  
	If $ w(V(\calR^4)) + w(C^4) + w(H^4) < 3p(p + \sqrt{p} (\lambda + 1))$, then $ |V(G)| \leq 3p(p + \sqrt{p} \pl)$, since $ V(\calR^4) \cup C^4 \cup H^4 \cup U' = V(G)$, where $ |U'| \leq 1$ by step~\ref{step::wVI2_b}.
	
	If we reach step~\ref{step::wVI3}, then $w(H^4_1) \leq 2p^{1.5}$, $|H^4| + |\calR^4| \leq p$ and $w(G) > 3p(p + \sqrt{p} (\lambda + 1))$.  
	Consequently, by \cref{cor::wVI1}, $H^5_1$ is non-empty, and the executed reduction is safe by \cref{lemma::saveRedwVI}, since $\lambda + 1 \leq \pl$.

	The calculation of a $\lambda$-BCD for $\lambda \in [p]$ runs in time $\O(|V|^3 |E|)$ by \cref{theorem:lccd}, as a $(\lambda + 1)$-packing has size less than $|V|$.  
	The extraction of a \dbe\ from it runs in time $\O(|V||E|)$ by \cref{thm::dbExapnsion}.  
	Each step consists of these constructions with an additional time cost of $\O(\log(p))$ in step~\ref{step::wVI2}.  
	As a result, the algorithm runs in time $\O(\log(p) |V|^3 |E|)$.
\end{proof}

%% file: W-separator-kernel-results.tex

In this section, we consider COC, a refined version of VI.
For VI, wVI, and wCOC, crown structures were integrated into the head and crown of the BCD.
In contrast, we now incorporate them into the body of the BCD.

Given a graph $G = (V, E)$ and parameters $k, W \in \mathbb{N}$, COC is the problem of removing at most $k$ vertices so that each resulting component has at most $W$ vertices.
We present a kernelization algorithm (\normalfont{\alg2kw}) that produces a $2kW$-vertex kernel in FPT-runtime, parameterized by the size of a maximum $(W+1)$-packing, as stated in the following theorem.

\begin{theorem}
	\label{theorem::2kWKern}
	A vertex kernel of size $2kW$ for the  component order connectivity problem  can be computed in time $\O(r^3|V||E| \cdot r^{\min(3r,k)})$, where $r \leq k $ is the size of a maximum $(W+1)$-packing.
\end{theorem}

A kernel of size $2kW$ is also presented in~\cite{DBLP:conf/iwpec/KumarL16}, however with an FPT-runtime in the parameter $W$ and using linear programming methods.
In contrast, our result has an FPT-runtime in the parameter of a maximum $(W+1)$-packing $r \leq k$ and  is fully combinatorial.
Although generalizing our FPT-runtime algorithm to polynomial time remains challenging, our insights into the structural properties of crown decompositions suggest promising directions for future work.
Moreover, we demonstrate how our algorithm (\normalfont{\2alg2kw}) becomes polynomial for two cases: when $W = 1$ (Vertex Cover) and for claw-free graphs, as stated in the theorems below, whose proofs can be found in Section~\ref{sec:kernels for special COC}.

\begin{theorem}
	\label{thm::VC}
	For $W=1$, i.e.~for the vertex cover problem, algorithm \normalfont{\2alg2kw} provides a $2k$ vertex kernel in polynomial time.
\end{theorem}

\begin{theorem}
	\label{thm::COC}
	The component order connectivity problem admits a $2kW$ vertex-kernel on claw-free graphs in polynomial time.
\end{theorem}

We begin with preliminaries on the reducible structures of COC and their properties, which are essential for our proofs.
Following this, we focus on the kernelization algorithms and provide proofs of its correctness.

\subsection{Reducible Structures of the Component Order Connectivity Problem and its Properties}

We call a (not necessarily minimal) set of vertices that separates a graph into components of size at most $W$ a \emph{$W$-separator}. 
We define $(G, k, W)$ as an instance of COC, where it is labeled as a yes-instance if there exists a $W$-separator of size $k$; otherwise, it is labeled as a no-instance.
To design the kernelization algorithm, we briefly review, in addition to the preliminary content in \cref{sec::prelimCrownAndPack}, some definitions and theorems from \cite{DBLP:conf/gecco/Baguley0N0PZ23,DBLP:conf/esa/Casel0INZ21,fomin2019kernelization,DBLP:conf/iwpec/KumarL16}.
Of particular interest are the strictly reducible pairs introduced by Kumar and Lokshtanov~\cite{DBLP:conf/iwpec/KumarL16}, who proved that such structures must exist in graphs with more than $2kW$ vertices.

\begin{definition}[(strictly) reducible pair]
	\label{def::RedPair}
	For a graph $G=(V,E)$ and a parameter $W \in \N$, a pair $(A,B)$ of vertex disjoint subsets of $V$ is a \emph{reducible pair} for \normalfont{COC} if the following conditions are satisfied:
	\begin{itemize}
		\item $N(B) \subseteq A$.
		\item
		The size of each $Q \in \comp(G[B])$ is at most $W$.
		\item There is an assignment function $g \colon \comp(G[B]) \times A \to \N_0$, such that
		\begin{itemize}
			\item for all $Q \in \comp(G[B])$ and $a \in A$, if $g(Q,a) \ne 0$, then $a \in N(Q)$
			\item for all $a \in A$ we have $\sum_{Q \in \comp(G[B])} g(Q,a) \geq 2W-1$
			\item for all $Q \in \comp(G[B])$ we have  $\sum_{a \in A} g(Q,a) \leq |Q|$.
		\end{itemize} 
	\end{itemize} 
	In addition, if there exists an $a \in A$ such that $\sum_{Q \in \comp(G[B])} g(Q,a) \geq 2 W$, then $(A,B)$ is a \emph{strictly reducible pair}.
\end{definition}

We say that $(A, B)$ is a minimal (strictly) reducible pair if there does not exist a (strictly) reducible pair $(A', B')$ with $A' \subset A$ and $B' \subseteq B$. Clearly, if a graph $G$ contains a (strictly) reducible pair, then it also contains a minimal one. Furthermore, for $A' \subseteq A$, we define $\comp(G[B])_{A'} \subseteq \comp(G[B])$ as the set of components $\{ Q \in \comp(G[B]) \mid N(Q) \cap A' \neq \varnothing \}$. Before discussing the connection between a (strictly) reducible pair and COC, we state a property of minimal strictly reducible pairs found by Baguley et al.~\cite{DBLP:conf/gecco/Baguley0N0PZ23}.

\begin{lemma}[\cite{DBLP:conf/gecco/Baguley0N0PZ23}, Lemma~4.16]
	\label{lemma::SizeNeoghborhood}
	Let $(A,B)$ be a minimal strictly reducible pair in $G$ with parameter $W$.
	Then, for every $A' \subseteq A$ we have $|V(\comp(G[B])_{A'})| > |A'|(2W-1)$.
\end{lemma}

Considering \cref{lemma::SizeNeoghborhood}, if there is a minimal strictly reducible pair $(A, B)$ in $G$ and we remove a proper subset $A'$ of $A$ from $G$, then $G - A'$ still contains a strictly reducible pair with head vertices $A \setminus A'$. Specifically, the specified vertex, guaranteed to receive a fractional assignment of at least $2W$ over $g$ from $\comp(G[B])$, can be chosen arbitrarily (cf.~\cref{def::RedPair}). In our algorithm, these properties enable a stepwise approach in localizing minimal strictly reducible pairs.

A reducible pair $(A, B)$ is, in fact, a $(W, W)$-CD, or equivalently, an $(A, \varnothing, W, f, 2W - 1)$ \dbe\ where $B = f^{-1}(A)$. Note that $|B| > |A|(2W - 1)$; hence, it can be identified in graphs resembling bipartite graphs with the required expansion properties (see \cref{def::demBalDec} and \cref{thm::dbExapnsion}). Here, the head and crown correspond to $A$ and $B$, respectively. 

Once a $(W, W)$-CD is identified, we can reduce the COC instance by including the head vertices in the solution. In this case, no additional vertices from the crown are required, since after removing the head vertices, the crown’s connected components are isolated from $G$ and each connected component has size at most $W$. 

In the following, we outline why this reduction is safe for COC. Let $P_1, \dots, P_m \subseteq V$ be a $(W+1)$-packing. For a $W$-separator $S \subseteq V$, it holds that $S \cap P_i \neq \varnothing$ for all $i \in [m]$. Therefore, the size of a $(W+1)$-packing provides a lower bound on the number of vertices required for a $W$-separator. Considering a $(W, W)$-CD $C, H, f$ in $G$, the function $f$ enables the construction of a $(W+1)$-packing of size $|H|$ in $G[H \cup C]$. This essentially provides a lower bound of $|H|$ vertices for a $W$-separator in $G[H \cup C]$.
On the other hand, $H$ serves as a $W$-separator of $G[A \cup B]$ while separating $C$ from the rest of the graph. These properties support the following theorem for $(W, W)$-CDs.

\begin{lemma}[\cite{DBLP:conf/esa/Casel0INZ21}~Lemma 49]
	\label{thm::saveReduction}
	Let $(G,k)$ be an instance of COC and $(A,B,f)$ a $(W,W)$-CD in $G$.
	Then, $(G,k)$ is a yes-instance if and only if $(G - (A \cup B), k - |A|)$ is a yes-instance.
\end{lemma}

Finally, the relationship between the existence of a strictly reducible pair and the attainment of the desired kernelization of $2kW$ is demonstrated in the following lemma.

\begin{lemma}[\cite{fomin2019kernelization}, Lemma~6.14]
	\label{lemma::existenceReducible}
	Let $(G,k,W)$ be an instance of the $W$-separator problem, such that each component in $G$ has size at least $W+1$.
	If $|V| > 2kW$ and $(G,k,W)$ is a yes-instance, then there exists a strictly reducible pair $(A,B)$ in $G$. 
\end{lemma}

The fundamental structure utilized in the kernelization algorithm is again the BCD.
The established relationships to COC are presented in the following lemma.

\begin{lemma}[\cite{DBLP:conf/esa/Casel0INZ21}, Lemma 48 and 49]
	\label{lemma::saveReductionBCD}
	Let $(G,k,W)$ be an instance of COC, and $\chrf$ a $W$-BCD in $G$.
	$(G,k)$ is a yes-instance if and only if $(G - (H \cup C), k - |A|)$ is a yes-instance.
	Furthermore, if $|\calR| + |H| > k$, then $(G,k)$ is a no-instance.
\end{lemma}

\subsection{Kernelization Algorithm}

In essence, the kernelization algorithm focuses on analyzing pairs $(A,B)$ in a $W$-BCD.
We show that the head vertices $A$ have specific traits in a $W$-BCD, making it possible to locate them.
Structurally, the algorithm operates as a bounded search tree, hence it is presented recursively. 
The input is an instance $(G,k,W)$ of COC, where $|V(G)| > 2kW$ and each component of $G$ has at least $W+1$ vertices.
The size limits for the input graph are not arbitrary; if $|V(G)| \leq 2kW$, there is nothing to do, and removing components smaller than $W+1$ is a safe reduction.

For $V' \subseteq V$ we define $\sepw(V') \in \N$ as the cardinality of a minimum $W$-separator in $G[V']$, and $\argSepw(V') \subset V$ as an argument suitable to this cardinality.
For a graph $G' \subseteq G$ let $\gLarge(G')$ be the graph obtained by removing all components of size at most $W$ from $G$.
Lastly, for a $W$-BCD $\chrf$ we define $\calR':= \{R \in \calR \mid |R| > 2W \text{ and } \sepw(R)=1\}$ and $S_{\calR'} := \bigcup_{R \in \calR'} \argSepw(R)$
(uniqueness of $\argSepw(R)$ for $R \in \calR'$ is shown in \cref{lemma::UniqueSep}). 

\paragraph*{Find reducible structures (\alg2kw)}
\begin{enumerate}	
	\item
	Compute a $W$-BCD $\chrf$ in $G$ and initialize $t=|H|+|\calR|$ and $S = \varnothing$.
	
	\item 
	\label{step:recStart}
	
	Let $G' = \gLarge(G-S)$.
	If $G'$ is an empty graph return a trivial yes-instance.
	Otherwise, compute a $W$-BCD $\chrf$ in $G'$.
	\begin{enumerate}
		\item If $|H|+|\calR| > k$ return a trivial no-instance.
		\label{step::NoInst1}
		\item 
		\label{step::RedPairCeck}
		Let $\Q$ be the connected components of size at most $W$ in $G-S$.
		Let $A= S \cup H$ and $B= C \cup V(\Q)$.
		Compute a \dbe $(A_1,A_2,W,f,2W-1)$ in $G[A \cup B]$ using \cref{thm::dbExapnsion}.
		If $A_1 \neq \varnothing$, then terminate the algorithm and return $(G-(A_1 \cup f^{-1}(A_1)),k-|A_1|,W)$.
	\end{enumerate}
	
	\item If the depth of the recursion is more than $\min(3t,k)$, then break.
	\label{step::dephtOfRec}
	
	\item For each $v \in H \cup S_{\calR'}$:
	\label{step::forLoop}
	\begin{itemize}
		\item Add $v$ to $S$ and recurse from step~\ref{step:recStart}.
	\end{itemize}
	
	\item 
	\label{step::NoInst2}
	Return a trivial no-instance.
\end{enumerate}

We split the correctness proof of Algorithm \alg2kw into the following three lemmas:

\begin{lemma}
	\label{lemma::locateStrictPair}
	Let $(G=(V,E),k,W)$ be a yes-instance and let $(A,B)$ be a minimal strictly reducible pair in $G$.
	Then, Algorithm \alg2kw finds $(A,B)$ and has a runtime of $\O(r^2|V||E| \cdot r^{\min(3r,k)})$, where $r$ is the size of a maximum $(W+1)$-packing.
\end{lemma}

\begin{lemma}
	\label{lemma::saveRedInAlg1}
	In the Algorithm \alg2kw in step~\ref{step::RedPairCeck}, if the algorithm terminates, then $A_1$ is part of an optimal $W$-separator in $G$.
\end{lemma}

\begin{lemma}
	\label{lemma::conclusionIsCorrect}
	If Algorithm \alg2kw labels $(G,k,W)$ as no-instance or yes-instance, then this conclusion is correct.
\end{lemma}

If \cref{lemma::locateStrictPair} holds, then the instance can be reduced by algorithm \alg2kw, as guaranteed by \cref{thm::saveReduction} and \cref{lemma::existenceReducible}. As previously mentioned, a reducible pair $(A, B)$ is essentially a $(W, W)$-CD or equivalently an $(A, \varnothing, W, f, 2W - 1)$ \dbe. 

Note that since $(A, B)$ is a $(W, W)$-CD, it also forms a packing of size $|A|$ in $G$. Therefore, if algorithm \alg2kw functions correctly within the runtime specified by \cref{lemma::locateStrictPair}, this would provide the proof for \cref{theorem::2kWKern}, as the exhaustive application of algorithm \alg2kw is constrained by the size of a maximum $(W+1)$-packing.

We begin by showing that the potential reduction of the instance in step~\ref{step::RedPairCeck} is valid, thereby proving \cref{lemma::saveRedInAlg1}.

\begin{proof}[Proof of \cref{lemma::saveRedInAlg1}]
	Observe that for a \dbe $(A_1, A_2, W, f, 2W - 1)$ in $G$, the tuple $(A_1, V(f^{-1}(A_1)), f)$ forms a $(W, W)$-CD and thus constitutes a reducible structure, as established by \cref{thm::saveReduction}. 
	
	To demonstrate that the application of \cref{thm::dbExapnsion} is valid in step~\ref{step::RedPairCeck}, we must verify that the size constraints of the components in $B = C \cup V(\Q)$ are satisfied and that $A_1 \subseteq A = S \cup H$ separates $B_1 = V(f^{-1}(A_1)) \subseteq C \cup V(\Q)$ from the remainder of the graph.
	 
	First, note that $S$ separates $V(\Q)$ from the remainder of the graph, indicating that $C$ cannot have an edge to $V(\Q)$.
	Consequently, by the definitions of $\Q$ and $C$, the subgraph $G[B]$ contains only connected components of size at most $W$.
	Thus, the size constraints of the components in $B$ are satisfied for the application of \cref{thm::dbExapnsion}.

	Moreover, by the definitions of $H$ and $S$, the vertices $A = H \cup S$ separate $B$ from the rest of the graph. Specifically, since $A_1$ separates $B_1$ from $A \setminus A_1$ in $G[A \cup B]$ as stated by \cref{thm::dbExapnsion} (i.e., $N(B_1) \subseteq A_1$ in $G[A \cup B]$), the vertices $A_1$ also separate $B_1$ from $V(G) \setminus B_1$ in $G$.
\end{proof}

The key aspect of the algorithm is the localization of a minimal strictly reducible pair, if it exists in the graph, as stated in \cref{lemma::locateStrictPair}. Let $(A,B)$ be a minimal strictly reducible pair in $G$. As previously mentioned, Algorithm \alg2kw essentially functions as a bounded search tree with a working vertex set $S$, which consists of potential head vertices. If we can ensure that $S$ corresponds to the vertex set $A$ at any given point, then $(A,B)$ is localized; once this occurs, step~\ref{step::RedPairCeck} will guarantee its extraction.

By the definition of a strictly reducible pair, we have $|B| > (2W - 1)|A|$, while $A$ separates $B$ from the rest of the graph. Thus, \cref{thm::dbExapnsion}, with demands $d_a = 2W - 1$ for every $a \in A$, ensures the identification of $(A,B)$ as a reducible pair when $S = A$.

It is also noteworthy that a reducible pair other than $(A,B)$ may be identified in step~\ref{step::RedPairCeck}. However, as long as it fulfills the properties of such a pair (see \cref{def::RedPair}), the reduction remains valid (cf. \cref{lemma::saveRedInAlg1}).

The following lemma implies a depth-wise progression in the bounded search tree towards $S = A$. Let $(A,B)$ be a minimal strictly reducible pair in $G = (V,E)$ and let $S \subset V$. Let $\chrf$ be a $W$-BCD in $\gLarge(G - S)$.

\begin{lemma}
	\label{lemma::aIsAvailable}
	If $S \cap B = \varnothing$ and $S \nsubseteq A$, then $(H \cup S_{\calR'}) \cap A  \ne \varnothing$. 
\end{lemma}

If we reach step~\ref{step::forLoop}, we currently do not possess a reducible pair and can therefore assume that $S \neq A$. This indicates that \cref{lemma::aIsAvailable} provides the following relation to the algorithm: if $S \subset A$ (where $S$ may be empty), then we can find at least one vertex from $A$ in $H \cup S_{\calR'}$.
Since the algorithm expands $S$ in a branch by each vertex of $H \cup S_{\calR'}$, with at least one vertex originating from $A$, and repeats this process with the resulting vertex set, we will eventually arrive at the case $S = A$. 
Thus, if we can demonstrate that $S = A$ occurs in the worst case within the specified runtime of \cref{lemma::locateStrictPair}, then \cref{lemma::aIsAvailable} implies \cref{lemma::locateStrictPair}.

Towards proving \cref{lemma::aIsAvailable} we observe the following relation between minimal strictly reducible pairs and weighted crown decompositions.

\begin{lemma}
	\label{lemma::HeadEmptyIntersectionWithB}
	If $S \cap B = \varnothing$, then $H \cap B = \varnothing$. 
\end{lemma}
\begin{proof}
	
	Let $G' = \gLarge(G - S)$ and $C, H, \calR, f$ be a $W$-BCD in $G'$. The case where $A \subseteq S$ is trivial, since $A$ separates $B$ from the rest of the graph, while $\comp(G[B])$ consists solely of components with at most $W$ vertices. Thus, $H$ cannot contain vertices from $B$ because for every $h \in H$, the vertex set $V' = \{h\} \cup f^{-1}(h)$ has more than $W$ vertices, and $G[\{h\} \cup f^{-1}(h)]$ is connected. 
	
	Therefore we assume that $A \nsubseteq S$. Suppose, towards a contradiction, that $B' = H \cap B \neq \varnothing$.
	We will demonstrate that in this case, the constraints of the $W$-CD $C, H, f$ are violated.
	Each $b \in B'$ satisfies $|f^{-1}(b)| \geq W$, meaning that at least $W$ vertices are assigned to each $b$ over $f$, where $G[\{b\} \cup f^{-1}(b)]$ is connected and $\bigcap_{b \in B'} f^{-1}(b) = \varnothing$. 
	For each $b \in B'$, let $B_b$ denote the component of $\comp(G[B])$ that includes $b$, and define $A' = A \cap \bigcup_{b \in B'} f^{-1}(b) \subseteq C$. Since the components in $\comp(G[B])$ are of size at most $W$, the vertices $f^{-1}(b)$ intersect with at most $W-1$ vertices of $B_b$ for every $b \in B'$. Thus, each $f^{-1}(b)$ must contain at least one vertex $a \in A$, as $A$ separates $B$ from the rest of the graph. Consequently, we have $|A'| \geq |B'|$.
	
	Next, we demonstrate that $G[A' \cup V(\comp(G[B])_{A'})] - B'$ contains a connected subgraph of size at least $W + 1$. If this holds, then either $G'[C]$ has a component of size at least $W + 1$, or $C$ is not separated by $H$ in $G'$. Both scenarios contradict the fact that $(C,H,f)$ is a $(W,W)$-CD in $G'$. 
	
	The components $\comp(G[B])_{A'}$ reside in $G'$, as $a \in C \subset V(G')$ for every $a \in A'$, and $B \cap S = \varnothing$. Notice that $B_b \in \comp(G[B])_{A'}$ for every $b \in B'$, implying that at most $|B'|$ components of $\comp(G[B])_{A'}$ are separated by $B'$ in $G[A' \cup \comp(G[B])_{A'}] - B'$. 	
	By \cref{lemma::SizeNeoghborhood}, we have $|V(\comp(G[B])_{A'})| > |A'|(2W-1)$, where $|V(\comp(G[B])_{\{a\}})| > 2W-1$ for every $a \in A'$. Given that $|A'| \geq |B|$, due to the pigeonhole principle we conclude that there exists an $a' \in A'$ such that $|V(\comp(G[B])_{\{a'\}} \setminus \{B_b\}_{b \in B'})| > 2W - 1 - W = W - 1$. This follows from the inequality:
	$$ |V(\comp(G[B])_{A'}) \setminus \{B_b\}_{b \in B'}| > |A'|(2W - 1) - |B'|W \geq |A'|(2W - 1) - |A'|W = |A'|(W - 1)$$
	Note that $\{a'\} \cup V(\comp(G[B])_{\{a'\}} \setminus \{B_b\}_{b \in B'})$ induces a connected subgraph, as for every $Q \in \comp(G[B])_{\{a'\}}$, we have $a' \in N(Q)$. Consequently, $\{a'\} \cup V(\comp(G[B])_{\{a'\}} \setminus \{B_b\}_{b \in B'})$ has size at least $W + 1$ and induces a connected subgraph in $G[A' \cup \comp(G[B])_{A'}] - B' \subseteq G'$.

\end{proof}

The next lemma we use to prove \cref{lemma::aIsAvailable} is that for a $W$-BCD $\chrf$ the $W$-separator $\argSepw(R)$ is unique for each $R \in \calR'$.

\begin{lemma}
	\label{lemma::UniqueSep}
	Let $G = (V,E)$ be a graph and let $W$ be a positive integer.
	If $|V(G)| > 2W$ and $G$ contains a $W$-separator $\{s\}$, then $s$ is the only $W$-separator in $G$.   
\end{lemma}
\begin{proof}
	Let $\calV$ denote the components of $G - \{s\}$. Suppose, towards a contradiction, that there exists another $W$-separator $s' \neq s$ in $G$. This implies that $s'$ belongs to some $V' \in \calV$ where $|V'| \leq W$. Consequently, we have $|V \setminus V'| = |V| - |V'| > 2W - W = W$. Furthermore, by the definition of $s$, every vertex set in $\calV$ is connected to $s$, so $G[V \setminus V']$ is connected. This shows that $s' \in V'$ cannot be a $W$-separator in $G$.
\end{proof}

With \cref{lemma::HeadEmptyIntersectionWithB,lemma::UniqueSep} in hand, we are now prepared to prove \cref{lemma::aIsAvailable} and, using \cref{lemma::aIsAvailable}, we can finally establish \cref{lemma::locateStrictPair}.

\begin{proof}[Proof of \cref{lemma::aIsAvailable}]
	Let $A' = A \setminus S$ and $B' = V(\comp(G[B])_{A'})$.
	If $A' \cap H \neq \varnothing$, then we are done.
	Therefore, we assume that $A' \cap H = \varnothing$ and prove for this case that $A' \cap S_{\calR'} \neq \varnothing$.	
	By \cref{lemma::SizeNeoghborhood}, we have $|V(\comp(G[B])_{A'})| > |A'|(2W-1)$, where $|V(\comp(G[B])_{\{a\}})| > 2W-1$ for every $a \in A'$.
	Thus, with $S \cap B = \varnothing$, the vertices $A'$ and $V(\comp(G[B])_{A'})$ must be in $\gLarge(G-S)$, since for every $a \in A'$, the subgraph $G[\{a\} \cup V(\comp(G[B])_{\{a\}})]$ is connected and of size greater than $W$.
	By \cref{lemma::HeadEmptyIntersectionWithB}, from $S \cap B = \varnothing$ it follows that $H \cap B = \varnothing$, and thus, by the precondition $A' \cap H = \varnothing$, every vertex of $B'$ must lie in some $R \in \calR$.
	Observe that if vertices of $B'$ are in an $R \in \calR$, then vertices of $A'$ must also lie in $R$.
	The reason for this is that $A'$ separates $B'$ from the rest of the graph in $G'$, while each component $\comp(G[B'])$ has size at most $W$, where $G[R]$ is connected and has size greater than $W$.
	Moreover, observe that by $|B'| = |V(\comp(G[B])_{A'})| > |A'|(2W-1)$ we have $|B'|/|A'| > 2W-1$.
	Combining these two facts, a simple average argument shows that there exists an $R' \in \calR$ such that $|B' \cap R'|/|A' \cap R'| \geq 2W$.
	Given that $|R| \leq 3W$ for each $R \in \calR$, this particular $R' \in \calR$ can contain only one $a' \in A'$.
	Fix such an $R' \in \calR$ and the corresponding vertex $a' \in A'$.
	It remains to show that $R' \in \calR'$ and $\{a'\} = \argSepw(R')$.
	Thus, we need to show that $|R'| > 2W$ and that $\{a'\}$ is a unique $W$-separator in $G[R']$. 
	For the former, by $|B' \cap R'|/|A' \cap R'| \geq 2W/1 = 2W$, the vertex set $R'$ contains at least $2W$ vertices from $B'$ in addition to $a'$.
	Hence, $|R'| > 2W$.
	For the latter, by \cref{lemma::UniqueSep}, it suffices to show that $a'$ is a $W$-separator in $G[R']$.
	Observe that because $a'$ is the only vertex of $A'$ in $R'$, it separates $R' \cap B$ from $R' \setminus (B \cup \{a'\})$ in $G[R']$, where each component $\comp(G[R' \cap B])$ has size at most $W$.
	On the other hand, due to $|R' \setminus (B' \cup \{a'\})| \leq 3W - (2W+1) = W-1$, the sizes of the remaining components $\comp(G[R' \setminus (B' \cup \{a'\})])$ are at most $W-1$.
	It follows that $\{a'\}$ is a $W$-separator in $G[R']$.
\end{proof}

\begin{proof}[Proof of \cref{lemma::locateStrictPair}]
	As previously explained, \cref{lemma::aIsAvailable} ensures a depth-wise progression towards $S = A$ within the bounded search tree. Once this is achieved, step~\ref{step::RedPairCeck} will identify $(A, B)$ as a reducible pair. It remains to demonstrate that this occurs within the specified runtime.
	
	Without loss of generality, we assume that $(A, B)$ is the sole reducible pair in $G$. Let $\chrf$ denote a $W$-BCD in $G$ and set $t = |H| + |\calR|$. Note that $\P = \calR \cup \{\{h\} \cup V(f^{-1}(h))\}_{h \in H}$ forms a $(W+1)$-packing of size $|H| + |\calR|$, and therefore $t$ is at most the size of an optimal $(W+1)$-packing, say $r$, which represents the branching factor of the bounded search tree.
	
	To bound the tree depth, we use the approximation result for a $(W+1)$-packing achieved through a $W$-BCD, as shown by Casel et al.~\cite{DBLP:conf/esa/Casel0INZ21} (Theorem 10). They proved that $\P$ corresponds to a 3-approximation for the $(W+1)$-packing problem, i.e., $r \leq 3t$.
	
	Moreover, a reducible pair corresponds to a $(W, W)$-CD, which is also a $(W+1)$-packing of size $|A| \leq r \leq 3t$ in $G$. Consequently, we can identify a reducible pair, if it exists, after reaching a depth of at most $3t \leq 3r$ in the bounded search tree.
	
	To determine the overall runtime, note that each step of Algorithm \alg2kw can be performed in time $\O(r^2|V||E|)$ (cf.~\cref{thm::dbExapnsion,theorem:lccd}). In each iteration of the algorithm, we compute a $W$-BCD $\chrf$, where the branching factor of the resulting search tree is at most $|\calR| + |H| \leq t \leq r$. The recursion depth is limited by $\min(k, 3t) \leq 3r$ (cf.~step~\ref{step::dephtOfRec}).
	
	Since $(G, k, W)$ is a yes-instance, we can assume $|A| \leq 3t \leq 3r$, and therefore we reach the case $S = A$ within a runtime of $\O(r^2|V||E| \cdot r^{\min(3r, k)})$.
\end{proof}

Finally, we prove \cref{lemma::conclusionIsCorrect}, which concludes the proof of the correctness of the algorithm \alg2kw and, thereby, also proves \cref{theorem::2kWKern}.

\begin{proof}[Proof of \cref{lemma::conclusionIsCorrect}]
	If the algorithm ends with the conclusion that $(G,k,W)$ is a yes-instance in step~\ref{step:recStart}, then $G-S$ only contains components of size at most $W$. It follows that $S$ is a $W$-separator. Since the depth of the recursion is at most $k$ (cf.~step~\ref{step::dephtOfRec}), the set $S$ has a cardinality of at most $k$ and therefore $(G,k,W)$ is a yes-instance.
	
	If the algorithm terminates in step~\ref{step::NoInst1}, then by \cref{lemma::saveReductionBCD} the decision that $(G,k,W)$ is a no-instance is correct. Otherwise, the algorithm terminates in step~\ref{step::NoInst2}. Since $|V(G)| > 2kW$ and every component of $G$ has size at least $W+1$, according to \cref{lemma::existenceReducible} if $(G,k,W)$ is a yes-instance there must be a strictly reducible pair. However, \cref{lemma::locateStrictPair} guarantees that we would then also find it. It follows that $(G,k,W)$ must be a no-instance.
\end{proof}

\subsection{Polynomial Time Kernels for Special Cases of COC}\label{sec:kernels for special COC}

Towards proving \cref{thm::VC,thm::COC} we introduce another algorithm designed to compute a $2kW$ vertex kernel for COC.
This algorithm shares the same principles as \alg2kw but exhibits polynomial running times for specific cases. 

\paragraph*{Find reducible structures (\2alg2kw):}

\begin{enumerate}	
	\item
	Initialize $S = \varnothing$ and $G'=G$.
	
	\item While $G' \neq (\varnothing, \varnothing)$:
	\begin{enumerate}
		\item
		\label{step::noVertexOfB}		
		Compute a $W$-BCD $\chrf$ in $G'$, such that for a minimal strictly reducible pair $(A,B)$ in $G$ we have $B \cap S_{\calR'}= \varnothing$.
		
		\item
		\label{step::noInstAlg2}
		If $|H|+|\calR| > k$, or $\calR' = \varnothing$ and $H=\varnothing$, then terminate the while-loop.
		\item
		Add the vertices $S_{\calR'}$ and $H$ to $S$. 
		\item
		\label{step::saveRedAlg2}
		Let $\Q$ be the connected components of size at most $W$ in $G-S$.
		Let $A= S \cup H$, $B= C \cup V(\Q)$. 
		Compute a \dbe $(A_1,A_2,W,f,2W-1)$ in $G[A \cup B]$ by using \cref{thm::dbExapnsion}.
		If $A_1 \neq \varnothing$, then terminate the algorithm and return $(G-(A_1 \cup f^{-1}(A_1)),k-|A_1|,W)$.
		
		\item
		Update $G'$ by $\gLarge(G-S)$.
	\end{enumerate} 
	
	\item Return a trivial no-instance.
\end{enumerate}

The crucial difference to Algorithm \alg2kw is step~\ref{step::noVertexOfB}.
Unfortunately, we have not succeeded in finding a polynomial algorithm in general for this step.
However, it has polynomial runtimes when $W = 1$ (Vertex Cover) and for claw-free graphs.

For proving that Algorithm \2alg2kw works correctly, note that step~\ref{step::saveRedAlg2} in Algorithm \2alg2kw is identical to step~\ref{step::RedPairCeck} in \alg2kw.
If the algorithm terminates at this step, then the reduction is safe by \cref{lemma::saveRedInAlg1}, i.e., $A_1$ is part of an optimal solution.
The remaining points for the correctness proof are stated in the following two lemmas:

\begin{lemma}
	\label{lemma::locateStrictPairAlg2}
	Let $(G,k,W)$ be the input of Algorithm \2alg2kw and let $(A,B)$ be a minimal strictly reducible pair in $G$.
	Then, Algorithm \2alg2kw finds $(A,B)$ after at most $|A|$ iterations of the while-loop, unless another reducible pair is localized in step~\ref{step::saveRedAlg2} between these steps.
\end{lemma}
\begin{lemma}
	\label{lemma::Alg2ConclusionCorrect}
	If Algorithm \2alg2kw labels $(G,k,W)$ as a no-instance, then this conclusion is correct.
\end{lemma}

Regarding \cref{lemma::locateStrictPairAlg2}, note that if we identify a reducible pair different from $(A,B)$ in step~\ref{step::saveRedAlg2}, this still progresses towards a $2kW$ vertex kernel, as the reduction is safe in any case.

\begin{proof}[Proof of \cref{lemma::locateStrictPairAlg2}]
	Without loss of generality, we assume that $(A,B)$ is the only reducible pair. We prove that $A \subseteq S$ and $S \cap B = \varnothing$ after at most $|A|$ iterations of the while-loop. By \cref{thm::dbExapnsion}, step~\ref{step::saveRedAlg2} would then identify $(A,B)$.
	
	Since at each step in the while-loop we have $B \cap S_{\calR'} = \varnothing$, we obtain by \cref{lemma::HeadEmptyIntersectionWithB} that $H \cap B = \varnothing$, which in turn ensures that $S$ never contains a vertex of $B$. If $A \nsubseteq S$ at any iteration of the while-loop, then by \cref{lemma::aIsAvailable} we have $\left(S_{\calR'} \cup H\right) \cap A \neq \varnothing$, which implies that at least one vertex of $A$ will be added to $S$. Consequently, we have $A \subseteq S$ and $S \cap B = \varnothing$ after at most $|A|$ iterations of the while-loop.
\end{proof}

\begin{proof}[Proof of \cref{lemma::Alg2ConclusionCorrect}]
	There are two possible ways that Algorithm \2alg2kw labels $(G,k,W)$ as a no-instance. The first case is when $|H| + |\calR|$ exceeds $k$ in an iteration (cf.~step~\ref{step::noInstAlg2}), which is a correct conclusion by \cref{lemma::saveReductionBCD}. The second case occurs if the algorithm does not find a reducible pair, either because $S$ cannot be enlarged (cf.~step~\ref{step::noInstAlg2}) or the while-loop terminates when $G'$ becomes empty. Since $|V(G)| > 2kW$ and every component of $G$ has size at least $W+1$, \cref{lemma::existenceReducible} implies that, if $(G,k,W)$ is a yes-instance, there must be a strictly reducible pair $(A,B)$. However, \cref{lemma::locateStrictPairAlg2} ensures that we would identify such a pair within at most $|A|$ iterations of the while-loop. Therefore, $(G,k,W)$ must be a no-instance.
\end{proof}

\subsubsection*{Vertex Cover}

To introduce a novel approach for computing a $2k$-vertex kernel for the vertex cover problem, we present a result that establishes a $2kW$-vertex kernel for the COC problem, applicable for any integer $W$. However, the associated algorithm achieves polynomial runtime only when $W=1$. This limitation arises because computing a maximum $(W+1)$-packing is only polynomially solvable for $W=1$, where it corresponds to finding a maximum matching.
In this section we prove the following theorem.

\begin{theorem}
	\label{thm::2kVC}
	In the Algorithm \2alg2kw, if in step~\ref{step::noVertexOfB} the value $|H| + |\calR|$ is the size of a $(W+1)$-packing in $G'$ in each iteration of the while-loop, then the algorithm works correctly.
\end{theorem}

To understand how we can apply the algorithm \2alg2kw (which essentially ensures that $B \cap S_{\calR'} = \varnothing$ in each iteration of the while-loop), we need to extend the algorithm to compute a $W$-BCD $\chrf$.
First, observe that $\P = \calR \cup { {h} \cup V(f^{-1})(h)}_{h \in H}$ forms a $(W+1)$-packing of size $|H| + |\calR|$. In the algorithm for computing a $W$-BCD, we define two corresponding sets, $\calR'$ and $H'$, that always maintain a $(W+1)$-packing of size $|H'| + |\calR'|$.
These sets, in particular, measure progress in the sense that the size of the corresponding packing can only increase. A crucial aspect here is that the algorithm can start with any arbitrary maximal $(W+1)$-packing. If this initial packing is a maximum $(W+1)$-packing, then upon termination, $\P$ will also correspond to a maximum $(W+1)$-packing.
In the original algorithm, $\calR'$ is initialized as the connected components of the input graph, with $H' = \varnothing$, but it is also possible to initialize $\calR'$ as a maximal $(W+1)$-packing instead (see~\cite{DBLP:conf/esa/Casel0INZ21}, proof sketch of Theorem~7).

Our first step is to establish a useful relation between a reducible pair and a maximum packing.

\begin{lemma}
	\label{lemma::maxPackAllAIndividual}
	Let $G$ be a graph, $W \in \N$ a parameter and $(A,B)$ a reducible pair in $G$ according to $W$.
	Then, any maximum $(W+1)$-packing $\P$ satisfies $|P \cap A| \leq 1$ for every $P \in \P$.
\end{lemma}

To prove \cref{lemma::maxPackAllAIndividual} we use the following lemma regarding reducible pairs. 

\begin{lemma}[\cite{DBLP:conf/iwpec/KumarL16}, Lemma 17]
	\label{lemma::packingReduciblePair}
	Let $(A,B)$ be a reducible pair in $G$.
	Then, there is a ($W+1$)-packing of size $|A|$ in $G[A \cup B]$.
\end{lemma}

\begin{proof}[Proof of \cref{lemma::maxPackAllAIndividual}]
	Suppose that $\P$ is a maximum $(W+1)$-packing in which there exists at least one $P \in \P$ with $|P \cap A| > 1$. Let $\P_A$ denote the elements in $\P$ that intersect with $A$, and define $\overline{\P}_A = \P \setminus \P_A$.
	Then, we have $|\P| = |\overline{\P}_A| + |\P_A| \leq |\overline{\P}_A| - |A| - 1$, as the elements in $\P_A$ are vertex-disjoint. To reach a contradiction that $\P$ is a maximum $(W+1)$-packing, we construct a larger $(W+1)$-packing from $\overline{\P}_A$ and $(A, B)$.
	Observe that for each $P \in \overline{\P}_A$, we have $P \cap B = \varnothing$ due to the properties of a reducible pair (cf.~\cref{def::RedPair}), whereby every connected subgraph of size $W+1$ containing a vertex of $B$ also contains a vertex of $A$. Thus, for each $P \in \overline{\P}_A$, we have $P \cap (A \cup B) = \varnothing$.
	Consequently, we can apply \cref{lemma::packingReduciblePair} to construct a $(W+1)$-packing $\P'$ of size $|A|$ within $G[A \cup B]$ that is disjoint from $\overline{\P}_A$. Therefore, $\P' \cup \overline{\P}_A$ forms a $(W+1)$-packing in $G$ of size $|\P'| + |\overline{\P}_A| = |A| + |\overline{\P}_A| > |A| - 1 + |\overline{\P}_A| \geq |\P|$.
\end{proof}

As a reminder, in Algorithm \2alg2kw, in each iteration, $\calR'$ represents the vertex sets in $\calR$ of the computed $W$-BCD $\chrf$ that contain a single $W$-separator and have a size larger than $2W$. The set $S_{\calR'}$ corresponds to the (unique) $W$-separators associated with these sets.

The essential condition we need to ensure in Algorithm \2alg2kw is that $S_{\calR'} \cap B = \varnothing$, as any overlap could prevent us from identifying a reducible pair $(A, B)$ if one exists in the graph. To demonstrate that this issue does not arise when we start the calculation of the $W$-BCD with a maximum $(W+1)$-packing, we make the following observation.

\begin{lemma}
	\label{lemma::bSepContains2A}
	Let $G$ be a graph, $W \in \N$ a parameter and $(A,B)$ a strictly reducible pair according to $W$.
	For any connected vertex set $R \subseteq V(G)$ with $|R| > 2W$ where a vertex $b \in B$ is a single $W$-separator in $G[R]$ we have $|A \cap B| > 1$.
\end{lemma}
\begin{proof}
	Let $\Q$ denote the components of $G[R] - {b}$. We will show that at least two elements of $\Q$ intersect with $A$, thereby proving the lemma.
	
	Define $B'$ as the component in $G[B]$ that contains $b$, and let $\Q'$ represent the components $Q \in \Q$ for which $B' \cap Q \neq \varnothing$ and $Q \cap A = \varnothing$. Since $A$ separates $B'$ from the rest of the graph, each component $Q \in \Q \setminus \Q'$ contains a vertex of $A$.
	
	Thus, if $|\Q \setminus \Q'| > 1$, the proof is complete. Note that $|V(\Q')| \leq |B' \setminus {b}| \leq W - 1$ and that $|Q| \leq W$ for each $Q \in \Q$. It follows that $|R \setminus (V(\Q') \cup {b})| > 2W - W = W$, which in turn implies that $|\Q \setminus \Q'| > 1$ since each element in $\Q$ has size at most $W$.
\end{proof}

By combining \cref{lemma::maxPackAllAIndividual} and \cref{lemma::bSepContains2A} and computing the $W$-BCD in step~\ref{step::noVertexOfB} using a maximum $(W+1)$-packing as a starting point, we establish \cref{thm::2kVC}. Since the case $W=1$ involves finding a maximum matching for this step, we thus prove \cref{thm::VC}.

\subsubsection*{Claw Free Graphs}

To demonstrate how \2alg2kw can be applied to claw-free graphs to obtain a $2kW$-vertex kernel in polynomial time, we first establish a more general theorem, as stated below. We assume that the input graph is connected; if not, each connected component is considered individually. Note that a strictly reducible pair always resides within a connected component of the input graph.

\begin{theorem}
	\label{thm::ClawFreeGeneral}
	In the Algorithm \2alg2kw, if $|\calR'| \leq 1$ and $H=\varnothing$ in step~\ref{step::noVertexOfB} in each iteration of the while-loop, then the algorithm works correctly.
\end{theorem}
\begin{proof}
	Let $(G,k,W)$ be an instance of COC and the input for Algorithm \2alg2kw, where $|V(G)| > 2kW$. If $G$ is a yes-instance, then, according to \cref{lemma::existenceReducible}, there must be a strictly reducible pair $(A,B)$ in $G$ with $|A| \leq k$. 
	
	The only requirement we need to establish is that $S_{\calR'} \cap B = \varnothing$ in every iteration of the while-loop. If $(G,k,W)$ is a yes-instance, we can apply \cref{lemma::aIsAvailable} with $H = \varnothing$, which ensures that $S_{\calR'}$ is non-empty, and the vertices in $S_{\calR'}$ must belong to $A$. This implies that $S_{\calR'} \cap B = \varnothing$ in every iteration of the while-loop until $A \subseteq S$.
	
	Once we have $A \subseteq S$, \cref{thm::dbExapnsion} in step~\ref{step::saveRedAlg2} identifies $(A,B)$ as a strictly reducible pair, leading to the termination of the algorithm.   
\end{proof}

We now show how claw-free graphs can satisfy the precondition of \cref{thm::ClawFreeGeneral}. For a vertex-weighted graph $G = (V, E, w \colon V \to \mathbb{N})$, we define $w_{\max} := \max_{v \in V} w(v)$.

\begin{lemma}
	\label{lemma::clawFreeBCD}
	Let $G=(V,E,w \colon V \to \N)$ be a vertex weighted and connected claw-free graph.
	Let $W \in \N$ be a parameter, where $|V(G)| > W \geq \wmax$.
	Then, from $G$ we can construct a $W$-BCD $\chrf$ in polynomial time, such that $C,H=\varnothing$ and $\calR$ contains at most one set $R \in \calR$ of size larger than $2W$.
\end{lemma}

To prove \cref{lemma::clawFreeBCD}, we use the connected vertex partition results from Borndörfer et al.~\cite{DBLP:conf/approx/BorndorferCINSZ21} for claw-free graphs. These results cannot be applied directly, as we work with a slightly weaker condition that yields a better outcome, which we explain shortly. Nevertheless, the subsequent proofs remain identical. The first property of claw-free graphs that we use is given in the following lemma.

\begin{lemma}
	\label{lemma:bounded_degree}
	Let $G$ be a claw-free graph and $T_r$ a DFS-tree rooted at $r \in V(G)$.
	Each vertex of $T_r$ has at most $2$ child vertices. 
\end{lemma}
\begin{proof}
	Assume that $T_r$ has a vertex $v$ with $m > 2$ child vertices $v_1, \dots, v_m$. Since $T_r$ is a DFS-tree, we know that $v_i v_j \notin E(G)$ for $i \neq j$. Thus, without loss of generality, the induced subgraph $G[{v, v_1, v_2, v_3}]$ forms a claw, which contradicts the assumption that $G$ is claw-free.
\end{proof}

With \cref{lemma:bounded_degree} established, we can remove a connected vertex set $S$ with weight between $W+1$ and $2W$ without disconnecting the graph, ensuring that the remaining graph remains claw-free.
The difference in the precondition of the following lemma, compared to \cite{DBLP:conf/approx/BorndorferCINSZ21}, lies in our assumption that $\wmax < W+1$ rather than $\wmax \leq W+1$. 
This enables $w(S) \leq 2W$ instead of $w(S) \leq 2W+1$.

\begin{lemma}
	\label{lemma::findConSub}
	Let $G=(V,E,w \colon V \to \N)$ be a claw-free, vertex weighted and connected graph.
	Let $W \in \N$ be a parameter, where $w(V(G)) \geq W+1 > w_{\max}$.
	There is a polynomial time algorithm that finds a connected vertex set $S$, such that $W+1 \leq w(S) \leq 2W$ and
	$G-S$ is connected. 
\end{lemma}
\begin{proof}
	Let $T_r$ be a DFS-tree in $G$ rooted at $r \in V$. For a vertex $v \in V$, define $T_v$ as the subtree rooted at $v$. By \cref{lemma:bounded_degree}, every vertex in $T_r$ has at most two child vertices. Since $w(V(G)) \geq W+1$, there exists a vertex $v$ with at most two child vertices $v_1$ and $v_2$ in $T_r$ such that $w(T_v) \geq \lambda$ while $w(T_{v_i}) < \lambda$. Such a vertex can be easily found by a bottom-up traversal from the leaves in polynomial time. 
	
	If $W+1 \leq w(T_v) \leq 2W$, we set $S=V(T_v)$ and are done. Otherwise, if $w(T_v) > 2W$, a simple calculation shows that $v$ must have two child vertices, as $w(v) \leq W$ and $w(T_{v_i}) \leq W$. 
	
	If $v=r$, we choose $S = \{r\} \cup V(T_{v_1})$. The induced subgraph of $S$ is connected with 
	$w(S) = w(v) + w(T_{v_1}) < w_{\max} + W \leq 2W$ and
	$w(S) = w(T_v) - w(T_{v_2}) > 2W - W = W$.
	Moreover, $G - S$ is connected.
	
	Now, consider $v \neq r$, and let $u$ be the parent of $v$ in $T_r$. Then $T_r[\{u, v, v_1, v_2\}]$ forms a claw, i.e., a $K_{1,3}$. Since $G$ is claw-free, there exists an edge $uv_j \in E \setminus E(T_r)$ for at least one $j \in \{1, 2\}$. Without loss of generality, let $j = 2$. We set $S = V(T_v) \setminus V(T_{v_2})$ as a connected vertex set and obtain, analogous to the case $v = r$, the desired weight conditions for $S$. Finally, the existence of $uv_2 \in E$ ensures that $G - S$ is connected.
\end{proof}

Using \cref{lemma::findConSub}, we can prove \cref{lemma::clawFreeBCD}.
Together with this result and \cref{thm::ClawFreeGeneral}, we conclude the proof of \cref{thm::COC}.

\begin{proof}[Proof of \cref{lemma::clawFreeBCD}]
	First, observe that a claw-free graph remains claw-free after vertex deletion. Therefore, we can exhaustively apply \cref{lemma::findConSub} with $\lambda = W+1$ to obtain a connected vertex partition $\mathcal S = \{S_1, \dots, S_m\}$, where $W+1 \leq w(S_i) \leq 2W$ for $i \in [m-1]$, and either $W+1 \leq w(S_m) \leq 2W$ or $w(S_m) \leq W$. 
	
	In the former case, we can construct the desired $W$-BCD with $C, H = \varnothing$ and $\calR = \mathcal S$. In the latter case, $S_m$ must be connected to some $S_j$ for $j \in [m-1]$. We merge $S_m$ and $S_j$ into a single connected vertex set, denoted $S'$. Note that $w(S') = w(S_j) + w(S_m) \leq 3W$. Consequently, setting $C, H = \varnothing$ and $\calR = \left(\mathcal S \setminus \{S_j, S_m\}\right) \cup \{S'\}$ yields the desired $W$-BCD.
\end{proof}

%% file: journalappendix.tex
The proof of Theorem~\ref{thm::dbExapnsion} can be derived by slightly modifying the approach for computing a balanced expansion, defined as a \dbe $(A_1, A_2, y, f, q)$ for $q \in \N$, in the work of Casel et al.~\cite{DBLP:conf/esa/Casel0INZ21} (Theorem 2). The full proof is provided in the arXiv version~\cite{casel2020balanced}. The only variation here is the allowance for different demands on $A$. Specifically, we need to address the proof of Lemma 4 in \cite{casel2020balanced}, which represents the intermediate stage of a balanced expansion, termed fractional balanced expansion, where components in $B$ may be partially assigned to $A$. 

We first present a generalized definition for a fractional demanded balanced expansion, followed by the necessary adjustments to Lemma 4 in \cite{casel2020balanced}.

\begin{definition}[fractional demanded balanced expansion] 
\label{def::demBalDec2}
For a graph $G=(A \cup B,E,w)$,
a partition $A_1 \cup A_2$ of $H$,
a function $g \colon A \times \comp(G[B]) \to H$,
demands $D=\{d_a\}_{a \in A}$ with $d_a \in \N$ for each $a \in A$
and $y \in \N$
the tuple $(A_1,A_2,y,f,D)$ is a demanded balanced expansion if
\begin{enumerate}
\item $w\left(a\right) + \sum_{Q \in \comp(G[B])} g\left(a,Q\right)\, \begin{cases} > \ d_a, & a\in A_1\\ 
\leq \ d_a, & a\in A_2\end{cases}$  
\item $\forall Q \in \comp(G[B]) \colon \sum_{a \in A} g\left(a,Q\right) \leq w\left(b\right)$ 
\item if $g(a,Q) \ne 0$, then $a \in N(Q)$ for $Q \in \comp(G[B])$ and $a \in A$,
\item $N(Q) \subseteq A_1$ for $Q \in \comp(G[B])$ if either $\sum_{a \in A_1} g(a,Q) > 0$, or $\sum_{a \in A} g(a,Q) < w(Q)$
\end{enumerate}
\end{definition}

\paragraph*{Modifications on Lemma 4 in \cite{casel2020balanced}}
The first step in the proof of Lemma 4 in the paper \cite{casel2020balanced} is to construct a flow network $N = \left(A \cup \comp(G[B]) \cup \left\{s,t\right\}, \overrightarrow{E}, c\right)$.
We do the same with respect to the demands $\{d_a\}_{a \in A}$.
While reading, we recommend a comparison with the illustration in \cref{pic::MaxFlowNetwork} (right).

\begin{figure}[H]
	\begin{center}
		\begin{tikzpicture}[scale=0.8,y=0.80pt, x=0.80pt, inner sep=0pt, outer sep=0pt]

			\path[draw=black, line width=0.25mm, -latex] (-100,15)--(-2,60);
			\path[draw=black, line width=0.25mm, -latex] (-100,15)--(-2,30);
			\path[draw=black, line width=0.25mm, -latex] (-100,15)--(-2,0);
			\path[draw=black, line width=0.25mm, -latex] (-100,15)--(-2,-30);
			\path[draw=black, line width=0.25mm, -latex] (-100,15)--(-2,45);
			\path[draw=black, line width=0.25mm, -latex] (-100,15)--(-2,15);
			\path[draw=black, line width=0.25mm, -latex] (-100,15)--(-2,-15);

			\path[draw=black, line width=0.25mm, latex-] (198,21)--(102,60);
			\path[draw=black, line width=0.25mm, latex-] (197,18)--(102,30);
			\path[draw=black, line width=0.25mm, latex-] (197,15)--(102,0);
			\path[draw=black, line width=0.25mm, latex-] (198,12)--(102,-30);

			\path[draw=black, line width=0.25mm, -latex] (0,60)--(98,60);	
			\path[draw=black, line width=0.25mm, -latex] (0,45)--(97,60);			
			\path[draw=black, line width=0.25mm, -latex] (0,30)--(97,59);
			\path[draw=black, line width=0.25mm, -latex] (0,30)--(98,30);
			\path[draw=black, line width=0.25mm, -latex] (0,15)--(98,30);
			\path[draw=black, line width=0.25mm, -latex] (0,0)--(98,-28);
			\path[draw=black, line width=0.25mm, -latex] (0,-30)--(98,58);
			\path[draw=black, line width=0.25mm, -latex] (0,-30)--(97,0);
			\path[draw=black, line width=0.25mm, -latex] (0,-15)--(98,30);
			\path[draw=black, line width=0.25mm, -latex] (0,-15)--(97,-31);

			\path[fill=black!60!, line width=0.25mm] (0,60) circle (0.09cm);
			\path[fill=black!60!, line width=0.25mm] (0,45) circle (0.09cm);
			\path[fill=black!60!, line width=0.25mm] (0,30) circle (0.09cm);
			\path[fill=black!60!, line width=0.25mm] (0,15) circle (0.09cm);
			\path[fill=black!60!, line width=0.25mm] (0,0) circle (0.09cm);
			\path[fill=black!60!, line width=0.25mm] (0,-15) circle (0.09cm);			
			\path[fill=black!60!, line width=0.25mm] (0,-30) circle (0.09cm);
			\path[fill=black, line width=0.25mm] (100,60) circle (0.09cm);
			\path[fill=black, line width=0.25mm] (100,30) circle (0.09cm);
			\path[fill=black, line width=0.25mm] (100,0) circle (0.09cm);			
			\path[fill=black, line width=0.25mm] (100,-30) circle (0.09cm);
			\path[fill=black, line width=0.25mm] (200,16) circle (0.09cm);
			\path[fill=black, line width=0.25mm] (-100,15) circle (0.09cm);	
			
			\node at (-100,30) {$s$};
			\node at (200,30) {$t$};
			\node at (0,75) {$B$};
			\node at (100,75) {$A$};
			\node at (-60,55) {$w(Q)$};
			\node at (170,60) {$d_a-w(a)$};
			\node at (50,-45) {$w(Q)$};
		\end{tikzpicture} 
		
	\end{center}
	\caption{Embedding of the graph in the corresponding flow network with capacities, where $Q \in \comp(G[B])$. }
	\label{pic::MaxFlowNetwork}
\end{figure}

We add the vertices $s,t$, $s$ as source and $t$ as sink, and arcs $\overrightarrow{E}$ with a capacity function $c\colon\overrightarrow{E} \to \mathbb{N}$ defined as follows:
connect every $Q \in \comp(G[B])$ through an arc $\overrightarrow{sQ}$ with capacity $|w(Q)|$ and connect every $a \in A$ through an arc $\overrightarrow{at}$ with capacity $d_a$.
As a note: The difference is at this point, in the original proof the capacity is $q \in \N$ in $\overrightarrow{at}$ for every $a \in A$.
Moreover, for every $Q \in \comp(G[B])$ and $a \in A$ add an arc $\ora{Qa}$ to $\ora{E}$ if $a \in N(Q)$ with capacity $|w(Q)|$.
The rest of the proof works identically.
As a note: 
By calculating different maximum $s$-$t$-flows in $N$, we generate the function $g$ with respect to $A_1$ and $A_2$.
The loss of $y = \max_{Q \in \comp(G[B])} w(Q)$ of the assignments via $g$ compared to the function $f$ in the required demanded balanced expansion results from the transition from a fractional demanded balanced expansion to a non-fractional one.

The last part of the theorem, i.e.~if there is an $A' \subseteq A$ with $w(A') + w(V(\B_{A'}) \geq \sum_{a \in A'} d_a$, then $A_1 \neq \varnothing$, originally, if $w(A) + w(B) \geq q|A|$, then $A_1 \neq \varnothing$, also holds for subsets of $A$, i.e.~if $w(A') + w(V(\B_{A'}) \geq q|A'|$, then $A_1 \neq \varnothing$, which is easily derivable.

%% file: sn-article.bbl

\begin{thebibliography}{28}
\ifx \bisbn   \undefined \def \bisbn  #1{ISBN #1}\fi
\ifx \binits  \undefined \def \binits#1{#1}\fi
\ifx \bauthor  \undefined \def \bauthor#1{#1}\fi
\ifx \batitle  \undefined \def \batitle#1{#1}\fi
\ifx \bjtitle  \undefined \def \bjtitle#1{#1}\fi
\ifx \bvolume  \undefined \def \bvolume#1{\textbf{#1}}\fi
\ifx \byear  \undefined \def \byear#1{#1}\fi
\ifx \bissue  \undefined \def \bissue#1{#1}\fi
\ifx \bfpage  \undefined \def \bfpage#1{#1}\fi
\ifx \blpage  \undefined \def \blpage #1{#1}\fi
\ifx \burl  \undefined \def \burl#1{\textsf{#1}}\fi
\ifx \doiurl  \undefined \def \doiurl#1{\url{https://doi.org/#1}}\fi
\ifx \betal  \undefined \def \betal{\textit{et al.}}\fi
\ifx \binstitute  \undefined \def \binstitute#1{#1}\fi
\ifx \binstitutionaled  \undefined \def \binstitutionaled#1{#1}\fi
\ifx \bctitle  \undefined \def \bctitle#1{#1}\fi
\ifx \beditor  \undefined \def \beditor#1{#1}\fi
\ifx \bpublisher  \undefined \def \bpublisher#1{#1}\fi
\ifx \bbtitle  \undefined \def \bbtitle#1{#1}\fi
\ifx \bedition  \undefined \def \bedition#1{#1}\fi
\ifx \bseriesno  \undefined \def \bseriesno#1{#1}\fi
\ifx \blocation  \undefined \def \blocation#1{#1}\fi
\ifx \bsertitle  \undefined \def \bsertitle#1{#1}\fi
\ifx \bsnm \undefined \def \bsnm#1{#1}\fi
\ifx \bsuffix \undefined \def \bsuffix#1{#1}\fi
\ifx \bparticle \undefined \def \bparticle#1{#1}\fi
\ifx \barticle \undefined \def \barticle#1{#1}\fi
\bibcommenthead
\ifx \bconfdate \undefined \def \bconfdate #1{#1}\fi
\ifx \botherref \undefined \def \botherref #1{#1}\fi
\ifx \url \undefined \def \url#1{\textsf{#1}}\fi
\ifx \bchapter \undefined \def \bchapter#1{#1}\fi
\ifx \bbook \undefined \def \bbook#1{#1}\fi
\ifx \bcomment \undefined \def \bcomment#1{#1}\fi
\ifx \oauthor \undefined \def \oauthor#1{#1}\fi
\ifx \citeauthoryear \undefined \def \citeauthoryear#1{#1}\fi
\ifx \endbibitem  \undefined \def \endbibitem {}\fi
\ifx \bconflocation  \undefined \def \bconflocation#1{#1}\fi
\ifx \arxivurl  \undefined \def \arxivurl#1{\textsf{#1}}\fi
\csname PreBibitemsHook\endcsname

\bibitem[\protect\citeauthoryear{Casel
  et~al.}{2021}]{DBLP:conf/esa/Casel0INZ21}
\begin{bchapter}
\bauthor{\bsnm{Casel}, \binits{K.}},
\bauthor{\bsnm{Friedrich}, \binits{T.}},
\bauthor{\bsnm{Issac}, \binits{D.}},
\bauthor{\bsnm{Niklanovits}, \binits{A.}},
\bauthor{\bsnm{Zeif}, \binits{Z.}}:
\bctitle{Balanced crown decomposition for connectivity constraints}.
In: \bbtitle{29th Annual European Symposium on Algorithms, {ESA} 2021,
  September 6-8, 2021, Lisbon, Portugal (Virtual Conference)}.
\bsertitle{LIPIcs},
vol. \bseriesno{204},
pp. \bfpage{26}--\blpage{12615}.
\bpublisher{Schloss Dagstuhl - Leibniz-Zentrum f{\"{u}}r Informatik},
  \blocation{}
(\byear{2021}).
\doiurl{10.4230/LIPIcs.ESA.2021.26} .
\burl{https://doi.org/10.4230/LIPIcs.ESA.2021.26}
\end{bchapter}
\endbibitem

\bibitem[\protect\citeauthoryear{Bentert
  et~al.}{2023}]{DBLP:conf/esa/BentertHK23}
\begin{bchapter}
\bauthor{\bsnm{Bentert}, \binits{M.}},
\bauthor{\bsnm{Heeger}, \binits{K.}},
\bauthor{\bsnm{Koana}, \binits{T.}}:
\bctitle{Fully polynomial-time algorithms parameterized by vertex integrity
  using fast matrix multiplication}.
In: \beditor{\bsnm{G{\o}rtz}, \binits{I.L.}},
\beditor{\bsnm{Farach{-}Colton}, \binits{M.}},
\beditor{\bsnm{Puglisi}, \binits{S.J.}},
\beditor{\bsnm{Herman}, \binits{G.}} (eds.)
\bbtitle{31st Annual European Symposium on Algorithms, {ESA} 2023, September
  4-6, 2023, Amsterdam, The Netherlands}.
\bsertitle{LIPIcs},
vol. \bseriesno{274},
pp. \bfpage{16}--\blpage{11616}.
\bpublisher{Schloss Dagstuhl - Leibniz-Zentrum f{\"{u}}r Informatik},
  \blocation{}
(\byear{2023}).
\doiurl{10.4230/LIPICS.ESA.2023.16} .
\burl{https://doi.org/10.4230/LIPIcs.ESA.2023.16}
\end{bchapter}
\endbibitem

\bibitem[\protect\citeauthoryear{Gima
  et~al.}{2022}]{DBLP:journals/tcs/GimaHKKO22}
\begin{barticle}
\bauthor{\bsnm{Gima}, \binits{T.}},
\bauthor{\bsnm{Hanaka}, \binits{T.}},
\bauthor{\bsnm{Kiyomi}, \binits{M.}},
\bauthor{\bsnm{Kobayashi}, \binits{Y.}},
\bauthor{\bsnm{Otachi}, \binits{Y.}}:
\batitle{Exploring the gap between treedepth and vertex cover through vertex
  integrity}.
\bjtitle{Theor. Comput. Sci.}
\bvolume{918},
\bfpage{60}--\blpage{76}
(\byear{2022})
\doiurl{10.1016/J.TCS.2022.03.021}
\end{barticle}
\endbibitem

\bibitem[\protect\citeauthoryear{Gima
  et~al.}{2024}]{DBLP:conf/walcom/GimaHKM0O24}
\begin{bchapter}
\bauthor{\bsnm{Gima}, \binits{T.}},
\bauthor{\bsnm{Hanaka}, \binits{T.}},
\bauthor{\bsnm{Kobayashi}, \binits{Y.}},
\bauthor{\bsnm{Murai}, \binits{R.}},
\bauthor{\bsnm{Ono}, \binits{H.}},
\bauthor{\bsnm{Otachi}, \binits{Y.}}:
\bctitle{Structural parameterizations of vertex integrity}.
In: \beditor{\bsnm{Uehara}, \binits{R.}},
\beditor{\bsnm{Yamanaka}, \binits{K.}},
\beditor{\bsnm{Yen}, \binits{H.}} (eds.)
\bbtitle{{WALCOM:} Algorithms and Computation - 18th International Conference
  and Workshops on Algorithms and Computation, {WALCOM} 2024, Kanazawa, Japan,
  March 18-20, 2024, Proceedings}.
\bsertitle{Lecture Notes in Computer Science},
vol. \bseriesno{14549},
pp. \bfpage{406}--\blpage{420}.
\bpublisher{Springer}, \blocation{}
(\byear{2024}).
\doiurl{10.1007/978-981-97-0566-5\_29} .
\burl{https://doi.org/10.1007/978-981-97-0566-5\_29}
\end{bchapter}
\endbibitem

\bibitem[\protect\citeauthoryear{Hanaka
  et~al.}{2024}]{DBLP:journals/corr/abs-2402-09971}
\begin{botherref}
\oauthor{\bsnm{Hanaka}, \binits{T.}},
\oauthor{\bsnm{Lampis}, \binits{M.}},
\oauthor{\bsnm{Vasilakis}, \binits{M.}},
\oauthor{\bsnm{Yoshiwatari}, \binits{K.}}:
Parameterized vertex integrity revisited.
CoRR
\textbf{abs/2402.09971}
(2024)
\doiurl{10.48550/ARXIV.2402.09971}
{\href{https://arxiv.org/abs/2402.09971}{{2402.09971}}}
\end{botherref}
\endbibitem

\bibitem[\protect\citeauthoryear{Gima and
  Otachi}{2024}]{DBLP:journals/algorithmica/GimaO24}
\begin{barticle}
\bauthor{\bsnm{Gima}, \binits{T.}},
\bauthor{\bsnm{Otachi}, \binits{Y.}}:
\batitle{Extended {MSO} model checking via small vertex integrity}.
\bjtitle{Algorithmica}
\bvolume{86}(\bissue{1}),
\bfpage{147}--\blpage{170}
(\byear{2024})
\doiurl{10.1007/S00453-023-01161-9}
\end{barticle}
\endbibitem

\bibitem[\protect\citeauthoryear{Lampis and
  Mitsou}{2021}]{DBLP:conf/isaac/LampisM21}
\begin{bchapter}
\bauthor{\bsnm{Lampis}, \binits{M.}},
\bauthor{\bsnm{Mitsou}, \binits{V.}}:
\bctitle{Fine-grained meta-theorems for vertex integrity}.
In: \beditor{\bsnm{Ahn}, \binits{H.}},
\beditor{\bsnm{Sadakane}, \binits{K.}} (eds.)
\bbtitle{32nd International Symposium on Algorithms and Computation, {ISAAC}
  2021, December 6-8, 2021, Fukuoka, Japan}.
\bsertitle{LIPIcs},
vol. \bseriesno{212},
pp. \bfpage{34}--\blpage{13415}.
\bpublisher{Schloss Dagstuhl - Leibniz-Zentrum f{\"{u}}r Informatik},
  \blocation{}
(\byear{2021}).
\doiurl{10.4230/LIPICS.ISAAC.2021.34} .
\burl{https://doi.org/10.4230/LIPIcs.ISAAC.2021.34}
\end{bchapter}
\endbibitem

\bibitem[\protect\citeauthoryear{Barefoot
  et~al.}{1987}]{barefoot1987vulnerability}
\begin{barticle}
\bauthor{\bsnm{Barefoot}, \binits{C.A.}},
\bauthor{\bsnm{Entringer}, \binits{R.}},
\bauthor{\bsnm{Swart}, \binits{H.}}:
\batitle{Vulnerability in graphs-a comparative survey}.
\bjtitle{J. Combin. Math. Combin. Comput}
\bvolume{1}(\bissue{38}),
\bfpage{13}--\blpage{22}
(\byear{1987})
\end{barticle}
\endbibitem

\bibitem[\protect\citeauthoryear{Jacob
  et~al.}{2023}]{DBLP:journals/algorithms/JacobMR23}
\begin{barticle}
\bauthor{\bsnm{Jacob}, \binits{A.}},
\bauthor{\bsnm{Majumdar}, \binits{D.}},
\bauthor{\bsnm{Raman}, \binits{V.}}:
\batitle{Expansion lemma - variations and applications to polynomial-time
  preprocessing}.
\bjtitle{Algorithms}
\bvolume{16}(\bissue{3}),
\bfpage{144}
(\byear{2023})
\doiurl{10.3390/A16030144}
\end{barticle}
\endbibitem

\bibitem[\protect\citeauthoryear{Ganian
  et~al.}{2021}]{DBLP:journals/algorithmica/GanianOR21}
\begin{barticle}
\bauthor{\bsnm{Ganian}, \binits{R.}},
\bauthor{\bsnm{Ordyniak}, \binits{S.}},
\bauthor{\bsnm{Ramanujan}, \binits{M.S.}}:
\batitle{On structural parameterizations of the edge disjoint paths problem}.
\bjtitle{Algorithmica}
\bvolume{83}(\bissue{6}),
\bfpage{1605}--\blpage{1637}
(\byear{2021})
\doiurl{10.1007/S00453-020-00795-3}
\end{barticle}
\endbibitem

\bibitem[\protect\citeauthoryear{Bagga
  et~al.}{1992}]{DBLP:journals/dam/BaggaBGLP92}
\begin{barticle}
\bauthor{\bsnm{Bagga}, \binits{K.S.}},
\bauthor{\bsnm{Beineke}, \binits{L.W.}},
\bauthor{\bsnm{Goddard}, \binits{W.}},
\bauthor{\bsnm{Lipman}, \binits{M.J.}},
\bauthor{\bsnm{Pippert}, \binits{R.E.}}:
\batitle{A survey of integrity}.
\bjtitle{Discret. Appl. Math.}
\bvolume{37/38},
\bfpage{13}--\blpage{28}
(\byear{1992})
\doiurl{10.1016/0166-218X(92)90122-Q}
\end{barticle}
\endbibitem

\bibitem[\protect\citeauthoryear{Clark et~al.}{1987}]{clark1987computational}
\begin{barticle}
\bauthor{\bsnm{Clark}, \binits{L.H.}},
\bauthor{\bsnm{Entringer}, \binits{R.C.}},
\bauthor{\bsnm{Fellows}, \binits{M.R.}}:
\batitle{Computational complexity of integrity}.
\bjtitle{J. Combin. Math. Combin. Comput}
\bvolume{2},
\bfpage{179}--\blpage{191}
(\byear{1987})
\end{barticle}
\endbibitem

\bibitem[\protect\citeauthoryear{Drange
  et~al.}{2016}]{DBLP:journals/algorithmica/DrangeDH16}
\begin{barticle}
\bauthor{\bsnm{Drange}, \binits{P.G.}},
\bauthor{\bsnm{Dregi}, \binits{M.S.}},
\bauthor{\bsnm{Hof}, \binits{P.}}:
\batitle{On the computational complexity of vertex integrity and component
  order connectivity}.
\bjtitle{Algorithmica}
\bvolume{76}(\bissue{4}),
\bfpage{1181}--\blpage{1202}
(\byear{2016})
\doiurl{10.1007/s00453-016-0127-x}
\end{barticle}
\endbibitem

\bibitem[\protect\citeauthoryear{Kratsch
  et~al.}{1997}]{DBLP:journals/dam/KratschKM97}
\begin{barticle}
\bauthor{\bsnm{Kratsch}, \binits{D.}},
\bauthor{\bsnm{Kloks}, \binits{T.}},
\bauthor{\bsnm{M{\"{u}}ller}, \binits{H.}}:
\batitle{Measuring the vulnerability for classes of intersection graphs}.
\bjtitle{Discret. Appl. Math.}
\bvolume{77}(\bissue{3}),
\bfpage{259}--\blpage{270}
(\byear{1997})
\doiurl{10.1016/S0166-218X(96)00133-3}
\end{barticle}
\endbibitem

\bibitem[\protect\citeauthoryear{Li et~al.}{2008}]{DBLP:journals/ijcm/LiZZ08}
\begin{barticle}
\bauthor{\bsnm{Li}, \binits{Y.}},
\bauthor{\bsnm{Zhang}, \binits{S.}},
\bauthor{\bsnm{Zhang}, \binits{Q.}}:
\batitle{Vulnerability parameters of split graphs}.
\bjtitle{Int. J. Comput. Math.}
\bvolume{85}(\bissue{1}),
\bfpage{19}--\blpage{23}
(\byear{2008})
\doiurl{10.1080/00207160701365721}
\end{barticle}
\endbibitem

\bibitem[\protect\citeauthoryear{Fellows and
  Stueckle}{1989}]{fellows1989immersion}
\begin{barticle}
\bauthor{\bsnm{Fellows}, \binits{M.R.}},
\bauthor{\bsnm{Stueckle}, \binits{S.}}:
\batitle{The immersion order, forbidden subgraphs and the complexity of network
  integrity}.
\bjtitle{J. Combin. Math. Combin. Comput}
\bvolume{6}(\bissue{1}),
\bfpage{23}--\blpage{32}
(\byear{1989})
\end{barticle}
\endbibitem

\bibitem[\protect\citeauthoryear{Lee}{2019}]{DBLP:journals/corr/Lee16c}
\begin{barticle}
\bauthor{\bsnm{Lee}, \binits{E.}}:
\batitle{Partitioning a graph into small pieces with applications to path
  transversal}.
\bjtitle{Math. Program.}
\bvolume{177}(\bissue{1-2}),
\bfpage{1}--\blpage{19}
(\byear{2019})
\doiurl{10.1007/s10107-018-1255-7}
\end{barticle}
\endbibitem

\bibitem[\protect\citeauthoryear{Chen
  et~al.}{2019}]{DBLP:journals/tcs/ChenFSWY19}
\begin{barticle}
\bauthor{\bsnm{Chen}, \binits{J.}},
\bauthor{\bsnm{Fernau}, \binits{H.}},
\bauthor{\bsnm{Shaw}, \binits{P.}},
\bauthor{\bsnm{Wang}, \binits{J.}},
\bauthor{\bsnm{Yang}, \binits{Z.}}:
\batitle{Kernels for packing and covering problems}.
\bjtitle{Theor. Comput. Sci.}
\bvolume{790},
\bfpage{152}--\blpage{166}
(\byear{2019})
\doiurl{10.1016/J.TCS.2019.04.018}
\end{barticle}
\endbibitem

\bibitem[\protect\citeauthoryear{Xiao}{2017}]{DBLP:journals/jcss/Xiao17a}
\begin{barticle}
\bauthor{\bsnm{Xiao}, \binits{M.}}:
\batitle{Linear kernels for separating a graph into components of bounded
  size}.
\bjtitle{J. Comput. Syst. Sci.}
\bvolume{88},
\bfpage{260}--\blpage{270}
(\byear{2017})
\doiurl{10.1016/J.JCSS.2017.04.004}
\end{barticle}
\endbibitem

\bibitem[\protect\citeauthoryear{Kumar and
  Lokshtanov}{2016}]{DBLP:conf/iwpec/KumarL16}
\begin{bchapter}
\bauthor{\bsnm{Kumar}, \binits{M.}},
\bauthor{\bsnm{Lokshtanov}, \binits{D.}}:
\bctitle{A 2lk kernel for l-component order connectivity}.
In: \beditor{\bsnm{Guo}, \binits{J.}},
\beditor{\bsnm{Hermelin}, \binits{D.}} (eds.)
\bbtitle{11th International Symposium on Parameterized and Exact Computation,
  {IPEC} 2016, August 24-26, 2016, Aarhus, Denmark}.
\bsertitle{LIPIcs},
vol. \bseriesno{63},
pp. \bfpage{20}--\blpage{12014}.
\bpublisher{Schloss Dagstuhl - Leibniz-Zentrum f{\"{u}}r Informatik},
  \blocation{}
(\byear{2016}).
\doiurl{10.4230/LIPIcs.IPEC.2016.20} .
\burl{https://doi.org/10.4230/LIPIcs.IPEC.2016.20}
\end{bchapter}
\endbibitem

\bibitem[\protect\citeauthoryear{Fomin et~al.}{2019}]{fomin2019kernelization}
\begin{bbook}
\bauthor{\bsnm{Fomin}, \binits{F.V.}},
\bauthor{\bsnm{Lokshtanov}, \binits{D.}},
\bauthor{\bsnm{Saurabh}, \binits{S.}},
\bauthor{\bsnm{Zehavi}, \binits{M.}}:
\bbtitle{Kernelization: Theory of Parameterized Preprocessing}.
\bpublisher{Cambridge University Press}, \blocation{}
(\byear{2019}).
\doiurl{10.1017/9781107415157} .
\burl{https://doi.org/10.1017/9781107415157}
\end{bbook}
\endbibitem

\bibitem[\protect\citeauthoryear{Xiao and Kou}{2017}]{DBLP:conf/tamc/XiaoK17}
\begin{bchapter}
\bauthor{\bsnm{Xiao}, \binits{M.}},
\bauthor{\bsnm{Kou}, \binits{S.}}:
\bctitle{Kernelization and parameterized algorithms for 3-path vertex cover}.
In: \beditor{\bsnm{Gopal}, \binits{T.V.}},
\beditor{\bsnm{J{\"{a}}ger}, \binits{G.}},
\beditor{\bsnm{Steila}, \binits{S.}} (eds.)
\bbtitle{Theory and Applications of Models of Computation - 14th Annual
  Conference, {TAMC} 2017, Bern, Switzerland, April 20-22, 2017, Proceedings}.
\bsertitle{Lecture Notes in Computer Science},
vol. \bseriesno{10185},
pp. \bfpage{654}--\blpage{668}
(\byear{2017}).
\doiurl{10.1007/978-3-319-55911-7\_47} .
\burl{https://doi.org/10.1007/978-3-319-55911-7\_47}
\end{bchapter}
\endbibitem

\bibitem[\protect\citeauthoryear{Li and Zhu}{2018}]{DBLP:journals/tcs/LiZ18}
\begin{barticle}
\bauthor{\bsnm{Li}, \binits{W.}},
\bauthor{\bsnm{Zhu}, \binits{B.}}:
\batitle{A 2\emph{k}-kernelization algorithm for vertex cover based on crown
  decomposition}.
\bjtitle{Theor. Comput. Sci.}
\bvolume{739},
\bfpage{80}--\blpage{85}
(\byear{2018})
\doiurl{10.1016/J.TCS.2018.05.004}
\end{barticle}
\endbibitem

\bibitem[\protect\citeauthoryear{Minty}{1980}]{DBLP:journals/jct/Minty80}
\begin{barticle}
\bauthor{\bsnm{Minty}, \binits{G.J.}}:
\batitle{On maximal independent sets of vertices in claw-free graphs}.
\bjtitle{J. Comb. Theory {B}}
\bvolume{28}(\bissue{3}),
\bfpage{284}--\blpage{304}
(\byear{1980})
\doiurl{10.1016/0095-8956(80)90074-X}
\end{barticle}
\endbibitem

\bibitem[\protect\citeauthoryear{Downey and
  Fellows}{2012}]{downey2012parameterized}
\begin{bbook}
\bauthor{\bsnm{Downey}, \binits{R.G.}},
\bauthor{\bsnm{Fellows}, \binits{M.R.}}:
\bbtitle{Parameterized Complexity}.
\bpublisher{Springer}, \blocation{}
(\byear{2012}).
\burl{https://doi.org/10.1007/978-1-4612-0515-9}
\end{bbook}
\endbibitem

\bibitem[\protect\citeauthoryear{Baguley
  et~al.}{2023}]{DBLP:conf/gecco/Baguley0N0PZ23}
\begin{bchapter}
\bauthor{\bsnm{Baguley}, \binits{S.}},
\bauthor{\bsnm{Friedrich}, \binits{T.}},
\bauthor{\bsnm{Neumann}, \binits{A.}},
\bauthor{\bsnm{Neumann}, \binits{F.}},
\bauthor{\bsnm{Pappik}, \binits{M.}},
\bauthor{\bsnm{Zeif}, \binits{Z.}}:
\bctitle{Fixed parameter multi-objective evolutionary algorithms for the
  w-separator problem}.
In: \beditor{\bsnm{Silva}, \binits{S.}},
\beditor{\bsnm{Paquete}, \binits{L.}} (eds.)
\bbtitle{Proceedings of the Genetic and Evolutionary Computation Conference,
  {GECCO} 2023, Lisbon, Portugal, July 15-19, 2023},
pp. \bfpage{1537}--\blpage{1545}.
\bpublisher{{ACM}}, \blocation{}
(\byear{2023}).
\doiurl{10.1145/3583131.3590501} .
\burl{https://doi.org/10.1145/3583131.3590501}
\end{bchapter}
\endbibitem

\bibitem[\protect\citeauthoryear{Bornd{\"{o}}rfer
  et~al.}{2021}]{DBLP:conf/approx/BorndorferCINSZ21}
\begin{bchapter}
\bauthor{\bsnm{Bornd{\"{o}}rfer}, \binits{R.}},
\bauthor{\bsnm{Casel}, \binits{K.}},
\bauthor{\bsnm{Issac}, \binits{D.}},
\bauthor{\bsnm{Niklanovits}, \binits{A.}},
\bauthor{\bsnm{Schwartz}, \binits{S.}},
\bauthor{\bsnm{Zeif}, \binits{Z.}}:
\bctitle{Connected k-partition of k-connected graphs and c-claw-free graphs}.
In: \beditor{\bsnm{Wootters}, \binits{M.}},
\beditor{\bsnm{Sanit{\`{a}}}, \binits{L.}} (eds.)
\bbtitle{Approximation, Randomization, and Combinatorial Optimization.
  Algorithms and Techniques, {APPROX/RANDOM} 2021, August 16-18, 2021,
  University of Washington, Seattle, Washington, {USA} (Virtual Conference)}.
\bsertitle{LIPIcs},
vol. \bseriesno{207},
pp. \bfpage{27}--\blpage{12714}.
\bpublisher{Schloss Dagstuhl - Leibniz-Zentrum f{\"{u}}r Informatik},
  \blocation{}
(\byear{2021}).
\doiurl{10.4230/LIPICS.APPROX/RANDOM.2021.27} .
\burl{https://doi.org/10.4230/LIPIcs.APPROX/RANDOM.2021.27}
\end{bchapter}
\endbibitem

\bibitem[\protect\citeauthoryear{Casel et~al.}{2020}]{casel2020balanced}
\begin{botherref}
\oauthor{\bsnm{Casel}, \binits{K.}},
\oauthor{\bsnm{Friedrich}, \binits{T.}},
\oauthor{\bsnm{Issac}, \binits{D.}},
\oauthor{\bsnm{Niklanovits}, \binits{A.}},
\oauthor{\bsnm{Zeif}, \binits{Z.}}:
Balanced crown decomposition for connectivity constraints.
arXiv preprint arXiv:2011.04528
(2020)
\end{botherref}
\endbibitem

\end{thebibliography}
